\documentclass[11pt]{article}
\usepackage{bbm}
\usepackage{geometry}
\usepackage{color}
\usepackage{caption}
\usepackage{amsfonts}
\usepackage{latexsym, amssymb, amsmath, amscd, amsthm, amsxtra}
\usepackage{mathtools}
\usepackage{enumerate}
\usepackage[all]{xy}
\usepackage{mathrsfs}
\usepackage{fancyhdr}
\usepackage{listings}
\usepackage{hyperref}
\usepackage{makecell}

\def\P{\mathbb{P}}

\def\R{\mathbb{R}}

\def\11{\mathbbm{1}}
\def\Gc{\mathcal{G}}

\def\ER{Erd\H{o}s-R\'enyi\ }
\def\blue{\textcolor{blue}}

\def\K{\mathcal{K}}

\def\Qb{\mathbb Q}

\newcommand{\Pb}{\mathbb P}

\newtheorem{thm}{Theorem}[section]

\theoremstyle{definition}
\newtheorem{claim}[thm]{Claim}
\newtheorem{DEF}[thm]{Definition}

\newtheorem{proposition}[thm]{Proposition}

\newtheorem{lemma}[thm]{Lemma}

\newtheorem{remark}[thm]{Remark}

\numberwithin{equation}{section}

\begin{document}
	\title{Low-Degree Hardness of Detection for Correlated \ER Graphs}
	
	\author{Jian Ding\footnote{J. Ding is partially supported by NSFC Key Program Project No. 12231002.}\\Peking University \and Hang Du\\MIT \and Zhangsong Li\\Peking University}

	\maketitle

\begin{abstract}
    Given two \ER graphs with $n$ vertices whose edges are correlated through a latent vertex correspondence, we study complexity lower bounds for the associated correlation detection problem for the class of low-degree polynomial algorithms. We provide evidence that any degree-$O(\rho^{-1})$ polynomial algorithm fails for detection, where $\rho$ is the edge correlation. Furthermore, in the sparse regime where the edge density $q=n^{-1+o(1)}$, we provide evidence that any degree-$d$ polynomial algorithm fails for detection, as long as $\log d=o\big( \frac{\log n}{\log nq} \wedge \sqrt{\log n} \big)$ and the correlation $\rho<\sqrt{\alpha}$ where $\alpha\approx 0.338$ is the Otter's constant. Our result suggests that several state-of-the-art algorithms on correlation detection and exact matching recovery 
    may be essentially the best possible.
\end{abstract}

\section{Introduction and main result}

In this paper, we consider the correlated \ER graph model (as defined below) and we study the computational complexity lower bounds for the corresponding correlation detection problem.
For any integer $n$, denote by $\operatorname{U}_n$ the set of unordered pairs $(i,j)$ with $1\le i\neq j\le n$. 

\begin{DEF}[Correlated \ER graph model] {\label{def-correlated-random-graph}}
    Given an integer $n\ge 1$ and two parameters $p,s\in (0,1)$, for $(i,j) \in \operatorname{U}_n$ let $I_{i,j}$ be independent Bernoulli variables with parameter $p$, and let $J_{i,j}$ and $K_{i,j}$ be independent Bernoulli variables with parameter $s$. In addition, let $\pi_*$ be an independent uniform permutation on $[n]=\{1,\dots,n\}$. Then, we define a triple of correlated random graphs $(G,A,B)$ such that for $(i,j) \in \operatorname{U}_n$ (note that we identify a graph with its adjacency matrix)
    \[
G_{i,j}=I_{i,j},A_{i,j}=I_{i,j}J_{i,j},B_{i,j}=I_{i,j}K_{\pi_*(i),\pi_*(j)}\,.
    \]
    For ease of presentation, we shall reparameterize such that $q=ps$ and $\rho=\frac{s(1-p)}{1-ps}$ respectively. We denote the joint law of $(\pi_*,G,A,B)$ as $\Pb_{*,n,q,\rho}$, and the marginal law of $(A,B)$ as $\Pb_{n,q,\rho}$. 
\end{DEF}

Two basic problems regarding the correlated \ER graphs are as follows: (1) the detection problem, i.e., testing $\mathbb P_{n,q,\rho}$ against $\mathbb Q_{n,q}$ where $\mathbb Q_{n,q}$ is the law of two independent \ER graphs on $[n]$ with edge density $q$; (2) the matching problem, i.e., recovering the latent matching $\pi_*$ from $(A,B) \sim \P_{n,q,\rho}$. In this paper, we provide evidence for computational hardness on the detection problem by analyzing a specific class of algorithms known as \emph{low-degree polynomial algorithms}. Our main results are informally stated as below; see Theorems~\ref{thm-dense-regime} and \ref{thm-sparse-regime} for precise statements.
\begin{thm}[Informal]
    For sufficiently large integer $n$ and parameters $q,\rho\in (0,1)$, there is evidence suggesting that algorithms based on polynomials of degree $O(\rho^{-1})$ fail for detection in the correlated \ER graph model. 
    
    Furthermore, if $q,\rho$ satisfies $nq=n^{o(1)}$ and $\rho^2<{\alpha}-\varepsilon$ for some arbitrary constant $\varepsilon>0$ (where $\alpha \approx 0.338$ denotes the Otter's constant), then there is evidence suggesting that algorithms based on polynomials of degree $d$ fail for detection as long as
    \begin{equation}
             \label{eq-D(n,q)}
          \log d= o\left(\frac{\log n}{\log nq}\wedge \sqrt{\log n}\right)\,.   
    \end{equation}
\end{thm}

Given that there is standard reduction from detection to exact matching recovery (one can do correlation detection by running any exact matching algorithm and then checking the size of overlap, see e.g. \cite{WXY20}), our results also suggest that the (exact) graph matching problem is computationally hard in the aforementioned regimes. 
In particular, our results suggest that several of the state-of-the-art algorithms in the existing literature have nearly reached the limit of efficient algorithms, and thus give rise to a characterization of the computational threshold for detection and exact matching in the correlated \ER graph model. See Table~\ref{table:algorithm-vs-hardness} for a detailed comparison between efficient algorithms and the hardness results. 

It is also noteworthy that when $nq = n^{O((\log \log n)^{-1})}$, our result implies that as long as $\rho$ is below the threshold $\sqrt{\alpha}$, all polynomial-based algorithms with degree $(\log n)^{o(1)}$ fail for correlation detection. Remarkably, the situation would change drastically when one is equipped with the power of higher degree polynomials; it was shown in \cite{BCL+19} that whenever $q$ satisfies
$
n^{\Omega((\log \log n)^{-1})}\le nq\le n^{o(1)}$,
there are polynomial-based algorithms with degree $(\log n)^{O(1)}$ that can achieve detection and matching provided with $\rho \ge (\log n)^{-o(1)}$ (although their algorithms require pseudo-polynomial running time). In view of this, our upper bound for the degree in \eqref{eq-D(n,q)} is also tight in a certain sense.

\subsection{Backgrounds on the correlated random graph model}

The network alignment problem has rich connections from various applied fields such as social network analysis \cite{NS08,NS09}, computer
vision \cite{CSS06, BBM05}, computational biology \cite{SXB08, VCL+15} and natural language processing \cite{HNM05}. For theoretical study of network alignment, the correlated \ER graph is a simple and arguably a canonical probabilistic model and therefore has been extensively studied recently. On the one hand, we have a fairly complete understanding of the information thresholds for both the problems of detection and matching thanks to \cite{CK16, CK17, HM20, WXY20, WXY21, GML20, DD22a, DD22b}. On the other hand, progressively improved algorithms on both problems have been obtained; see \cite{PG11, YG13, LFP14, KHG15, FQRM+16, SGE17, BCL+19, DMWX21, FMWX22a, FMWX22b, BSH19, CKMP19, DCKG19, MX20, GM20, GML20, MRT21, MRT23, MWXY21+, GMS22+, MWXY23, DL22+, DL23+}. Notably, in a series of recent works \cite{GML20, GMS22+, MRT23, MWXY23, DL22+, DL23+}, various rather efficient algorithms have been proposed with provable guarantees on successes. In particular, when $q>\frac{\log n}{n}$, \cite{MWXY23} proposed an efficient algorithm for exact matching that succeeds when $\rho>\sqrt{\alpha}$; See also \cite{MWXY21+} (respectively, \cite{GML20,GMS22+}) for remarkable results on detection (respectively, partial recovery) of a similar flavor. Furthermore, in the dense regime when $q \geq n^{-1+\delta+o(1)}$ for some constant $\delta>0$, \cite{DL23+} proposed an efficient algorithm for exact matching (which also serves for the goal of detection) that succeeds as long as the correlation is non-vanishing.

Note that the lower bounds required on the correlation for the aforementioned state-of-the-art algorithms are substantially larger than the information thresholds. It is then natural to wonder whether this information-computation gap reflects the fundamental nature of the problem or merely reflects the lack of understanding by researchers. In light of this, the main goal of this work is to provide evidence supporting the former scenario (at least for detection and exact matching), i.e., the algorithms proposed in \cite{MWXY21+, MWXY23, DL23+} for correlation detection or exact matching may just essentially be the best possible. Since at this point it seems rather elusive to prove hardness on the graph matching problem for a typical instance under the assumption of P$\neq$NP, we can only hope to prove hardness under even (much) stronger hypothesis. For this purpose, the framework of low-degree polynomials emerges as a natural and compelling choice. 
Indeed, it has been proved that the class of low-degree polynomial algorithms is a useful proxy for computationally efficient algorithms, in the sense that the best-known polynomial-time algorithms for a wide variety of high-dimensional inference problems are captured by the low-degree class; see e.g. \cite{Hop18, SW22, KWB22}. Furthermore, it is conjectured in \cite{Hop18} that the failure of degree-$d$ polynomial algorithms implies the failure of all ``robust'' algorithms with running time $n^{\Tilde{O}(d)}$ (here $\Tilde{O}$ means having at most this order up to a $\operatorname{poly} \log n$ factor).
We remark that the state-of-the-art algorithms for correlation detection and matching for correlated \ER graphs also belong to such class: the detection algorithm in \cite{MWXY21+} and the matching algorithm in \cite{MWXY23} are based on counting specific trees, and it can be generalized to counting graphs with bounded tree-width when the graph is dense. Also, the partial recovery algorithm in \cite{GML20} is a message-passing algorithm which exploits the local tree structure. In addition, although the algorithm proposed in \cite{DL23+} is not a low-degree polynomial algorithm $\emph{per se}$, it can be modified into a low-degree algorithm if one uses subgraph counts to replace the initialization procedure in \cite{DL23+}. 
As a result, it seems plausible that our impossibility result from low-degree polynomials captures the intrinsic average-case hardness of graph matching problems, or at least it is expected that breaking such impossibility results would require a major breakthrough in algorithms. With this belief in mind, our complexity lower bounds more or less match the state-of-the-art algorithms for correlation detection and exact matching, as summarized in the following table. Meanwhile, we do feel that the complexity lower bound in our result is also true for partial recovery (see developments on the algorithmic side in \cite{GML20, GMS22+} and predictions on hardness in \cite{PSS22, GMS22+}). But since there is no evident reduction from partial matching to detection, so far we do not have solid evidence on the computational hardness of partial matching. We believe such a hardness result requires a novel formulation of low-degree hardness framework for estimation, and we leave it for future work.

\begin{table}[ht!]
    \centering
    \begin{tabular}{|l|c|c|}
    \hline
    \ & Algorithms & Hardness \\
    \hline
    Detection & \thead{$nq\ge n^{\Omega(1)},\rho=\Omega(1);d=O(1)$ \cite{DL23+} \\ $nq\ge \Omega(1),\rho>\sqrt{\alpha};d=O(1)$ \cite{MWXY21+}} & \thead{$nq=n^{\Omega(1)},\rho=o(1);d=O(\rho^{-1})$\\ $nq=n^{o(1)},\rho<\sqrt{\alpha};\log d=o(D(n,q))$} \\
    \hline
    ExaMatch & \thead{$nq\ge n^{\Omega(1)}, \rho=\Omega(1); d=O(1)$ \cite{DL23+}\\ $nq\ge \log n,q>\sqrt{\alpha};d=O(\log n)$ \cite{MWXY23}} & \thead{$nq=n^{\Omega(1)},\rho=o(1);d=O(\rho^{-1})$\\ $nq=n^{o(1)},\rho<\sqrt{\alpha};\log d=o(D(n,q))$} \\
    \hline
    ParMatch & \thead{$nq=O(1),\rho>\sqrt{\alpha};d=O(\log n)$ \cite{GML20, GMS22+}}& Unknown \\
    \hline
    \end{tabular}
    \caption{Comparison between algorithms and hardness}
    \caption*{Let us further explain our convention regarding the triple $(n, q, d)$ in the table, and we take the first row for instance. In the column of ``Algorithms'', it means that when $nq \geq n^{\Omega(1)}$ and $\rho = \Omega(1)$, there exists a polynomial-based algorithm with $d = O(1)$ (and with polynomial running time) which succeeds on detection; in the column of ``Hardness'', it means that when $nq = n^{\Omega(1)}$ and $\rho = o(1)$, there is evidence suggesting that polynomial-based algorithms with degree below $d$  (for $d= O(\rho^{-1}$) fail for detection. The similar interpretations apply to other rows.  We write $D(n,q)$ for the quantity $\frac{\log n}{\log nq}\wedge \sqrt{\log n}$ as in \eqref{eq-D(n,q)}. 
    In addition, we expect that the ``Unknown'' block can be filled with similar results as above.}
        \label{table:algorithm-vs-hardness}
\end{table}

\subsection{The low-degree polynomial framework}
Drawing inspiration from the sum-of-squares hierarchy, the low-degree polynomial method offers a promising approach for establishing computational lower bounds in high-dimensional inference problems. In broad terms, this approach focuses on analyzing algorithms that rely on the evaluations of a collection of polynomials with moderate degrees. The impetus for delving into this category of algorithms stems from the examination of high-dimensional hypothesis testing problems \cite{BHK+19, HS17, HKP+17, Hop18}, with an extensive overview provided in \cite{KWB22}. Expanding upon this foundation, the low-degree framework has been subsequently applied to explore random optimization and constraint satisfaction problems.

As an appealing feature, the approach of low-degree polynomials has yielded tight hardness results for a range of problems, with prominent examples including but not limited to detection problems such as planted clique, planted dense subgraph, community detection, sparse-PCA (see \cite{HS17, HKP+17, Hop18, KWB22, SW22, DMW23+, BKW19, DKW19}), optimization problems such as maximal independent sets in sparse random graphs \cite{GJW20, Wein20}, and constraint satisfaction problems such as random $k$-SAT \cite{BH21}.
In the remaining of this paper, we will focus on applying this framework in the context of correlation detection for \ER graphs.


We first specify the class of algorithms we shall study.  For $n,d\ge 1$, denote $\mathcal P_{n,d}$ for the set of polynomials from $\{0,1\}^{2|\!\operatorname{U}_n\!|}$ to $\mathbb R$ with degree no more than $d$. 

\begin{DEF}[Polynomial algorithms for detection]\label{def-poly-test}
    For an algorithm $\mathcal A$ which tests $\mathbb P_{n, q,\rho}$ against $\mathbb Q_{n,q}$, we say it is a \emph{polynomial test with degree $d$} if the following holds:
    there exists a polynomial $f \in \mathcal{P}_{n,d}$ and some threshold $\tau\in \R$, such that $\mathcal A$ accepts $\Pb_{n,q,\rho}$ if and only if $f\big(\{A_{i,j}\}_{(i,j)\in \operatorname{U}_n},\{B_{i,j}\}_{(i,j)\in \operatorname{U}_n}\big)\ge \tau$.
\end{DEF}

We begin with the classic framework for assessing the computational complexity of polynomial tests, as raised in \cite{HS17}. In the remaining part of this paper, we will write $f=O(g)$ 
to indicate that $|f|\le Cg$ for some absolute constant $C>0$. 

\begin{DEF}{\label{def-evidence-hardness}}
    For a quadruple $(n,q,\rho,d)$ with $q,\rho,d$ possibly depending on $n$, we say that there is evidence that polynomial tests with degree at most $d$ fail for detection if the following holds: the signal-to-noise ratio is uniformly bounded for polynomial tests with $f\in \mathcal P_{n,d}$, i.e. as $n \to \infty$,
    \begin{equation}\label{eq-evidence-detection}
        \sup_{f\in \mathcal P_{n,d}}\frac{\mathbb E_{\Pb_{n,q,\rho}} [f]}{\sqrt{\mathbb E_{\Qb_{n,q}} [f^2}]}=O(1)\,.
    \end{equation}
\end{DEF}
\begin{remark}
This framework is closely related to the concept of strong separation introduced in \cite[Definition~1.8]{BAH+22}. To be specific, a real polynomial \(f\) defined on $$\big(\{A_{i,j}\}_{(i,j)\in \operatorname{U}_n},\{B_{i,j}\}_{(i,j)\in \operatorname{U}_n}\big)$$ is said to strongly separate \(\mathbb{P}_{n,q,\rho}\) and \(\mathbb{Q}_{n,q}\) if

\[
\sqrt{\operatorname{Var}_{\mathbb{P}_{n,q,\rho}}f(A,B) \vee \operatorname{Var}_{\mathbb{Q}_{n,q}}f(A,B)} = o\left(\left| \mathbb{E}_{\mathbb{P}_{n,q,\rho}}f(A,B) - \mathbb{E}_{\mathbb{Q}_{n,q}}f(A,B) \right|\right),\text{ as }n\to\infty.
\]

It can be deduced from Chebyshev's inequality that strong separation implies that the polynomial test based on $f$ with some appropriate threshold $\tau$ (see Definition~\ref{def-poly-test}) has vanishing type-I and type-II errors. Moreover, it is evident that if the optimal signal-to-noise ratio between \(\mathbb{P}_{n,q,\rho}\) and \(\mathbb{Q}_{n,q}\) remains bounded, no polynomial in \(\mathcal{P}_{n,d}\) can strongly separate these two measures. Consequently, \eqref{eq-evidence-detection} provides evidence of computational hardness in this context.
\end{remark}

Taking this framework, our first result indicates the computational hardness of detection up to degree $d=O(\rho^{-1})$ polynomials.

\begin{thm}\label{thm-dense-regime}
For any triple $(n,q,\rho)$, we have \eqref{eq-evidence-detection} holds for $d=O(\rho^{-1})$. This suggests that polynomial tests with degree at most $O(\rho^{-1})$ fail for detection. 
\end{thm}


Our next goal is to provide a more precise characterization of the computational threshold for correlation detection in the sparse regime where $q=n^{-1+o(1)}$, and particularly when \(nq=\operatorname{poly}\log n\). In this regime, the computational thresholds for both problems are expected to be at \(\rho=\sqrt{\alpha}\), where the upper bound has been (mostly) established in \cite{GML20, GMS22+, MWXY21+, MWXY23} and the lower bound is predicted by \cite{PSS22, GMS22+}. From the perspective of low-degree hardness for detection, it is tempting to argue that polynomial tests with moderate degrees are ineffective when \(\rho\) falls below this threshold. However, it is not hard to see that in this regime the previous framework cannot be applied directly, since some straightforward computations give that the left hand side of \eqref{eq-evidence-detection} is unbounded whenever $\rho d \to \infty$. That being said, this explosion does not necessarily indicate low computational complexity, since it is possible that the predominant contribution to the optimal signal-to-noise ratio originates from certain rare events. In order to address this, we note that the computational complexity for an instance sampled from a law $\mathbb P'$ should be statistically the same as that from $\mathbb P$, provided that $\mathrm{TV}(\mathbb P, \mathbb P') = o(1)$ (where $\mathrm{TV}(\mathbb P, \mathbb P')$ is the total variation distance between $\mathbb P$ and $\mathbb P'$). Therefore, we may choose such a $\mathbb P'$ under which those rare events with major contribution to the signal-to-noise ratio are excluded. This leads to our modified framework as follows. 
\begin{DEF}\label{def-framework-truncated-version}
    We say that there is evidence for the computational hardness of detection for polynomial tests with degree at most $d$, if \eqref{eq-evidence-detection} holds for $\Pb_{n,q,\rho}$ replaced by some $\Pb_{n,q,\rho}'$ with $\operatorname{TV}(\Pb_{n,q,\rho},\Pb_{n,q,\rho}')=o(1)$ as $n\to \infty$. 
    In addition, we refer to $\eqref{eq-evidence-detection}'$ the version of \eqref{eq-evidence-detection} after the aforementioned replacement. 
\end{DEF}
\begin{remark}\label{rmk-strong-separation}
The relationship between strong separation and Definition~\ref{def-framework-truncated-version} is as follows: according to \cite[Proposition~6.2]{BAH+22}, if the optimal signal-to-noise ratio of \(f\in \mathcal P_{n,d}\) between \(\mathbb{P}_{n,q,\rho}'\) and \(\mathbb{Q}_{n,q}\) remains bounded, then it implies that no polynomial in \(\mathcal P_{n,d}\) can strongly separate \(\mathbb{P}_{n,q,\rho}\) and \(\mathbb{Q}_{n,q}\), which provides evidence for the computational hardness of detection.
\end{remark}

The framework as in Definition~\ref{def-framework-truncated-version} enables us to demonstrate the low-degree hardness for detection in the correlated \ER graph model in the sparse regime, provided that $\rho$ is below the conjectured computational threshold.

\begin{thm}\label{thm-sparse-regime}
    For any fixed $\varepsilon>0$, assume $1 \le nq \le n^{o(1)}$ and $\rho^2<\alpha-\varepsilon$. 
    Then for some $\Pb_{n,q,\rho}'$ with $\operatorname{TV}(\Pb_{n,q,\rho},\Pb_{n,q,\rho}')=o(1)$, we have that $\eqref{eq-evidence-detection}'$ holds as long as the degree $d$ satisfies \begin{equation*}
        \log d=o\left(\frac{\log n}{{\log nq}}\wedge \sqrt{\log n}\right)\,.
    \end{equation*} As a result, this suggests the failure of low-degree polynomial tests. 
\end{thm}
We point out that in Theorem~\ref{thm-sparse-regime}, $\Pb_{n,q,\rho}'$ will be taken as the marginal probability of $\Pb_{*,n,q,\rho}'$, which denotes the conditional probability measure of $\Pb_{*,n,q,\rho}$ given that the mother graph $G$ does not contain any subgraph with unusually large density.
See Definition~\ref{def-addmisible} for precise statements.

In the subsequent sections of this paper, we will keep the values of \(n\), \(q\) and \(\rho\) fixed, and for the sake of simplicity we will omit subscripts involving these parameters. 
Throughout our discussions, we use $\operatorname{K}_n$ (in the normal font) to denote the complete graph on vertex set $[n]$ and the notation \(S\Subset\operatorname{K}_n\) denotes a graph $S$ within \([n]\) which has \emph{no isolated vertex}. The vertex set and edge set of \(S\) are represented as \(V(S)\) and \(E(S)\), respectively. For a permutation \(\pi\) on \([n]\), \(\pi(S)\) represents the graph with vertex set \(\{\pi(v): v \in V(S)\}\) and edge set \(\{(\pi(u), \pi(v)): (u, v) \in E(S)\}\). For \(S,T \Subset\operatorname{K}_n\), we say $S \subset T$ if $E(S) \subset E(T)$. Moreover, let \(S \cap T\Subset\operatorname{K}_n\) denote the graph induced by the edges in \(E(S) \cap E(T)\) (note that according to the definition the isolated vertices in $V(S)\cap V(T)$ should be removed in $S \cap T$). Define $S \cup T$, $S \setminus T$ and $S \triangle T$ in the similar manner.

Additionally, we define \(\mathcal{H}\) as the set of isomorphism classes of simple graphs with no isolated vertices. For each isomorphic class in \(\mathcal{H}\), we identify it with a specific graph \(\mathbf{H}\) in this class (we use the boldface font to indicate our emphasis on the isomorphism class). We denote by $\operatorname{Aut}(\mathbf H)$ the number of automorphisms of $\mathbf H$. In addition, for a set $A$, we denote by both $|A|$ and $\# A$ its cardinality.

\section{Proof of Theorem~\ref{thm-dense-regime}}\label{sec-proof-of-dense}
We start with the proof of Theorem~\ref{thm-dense-regime}, which is the relatively straightforward part of this paper. Although the calculation in this section is essentially the same as discussed in \cite[Section 2.3]{MWXY21+}, we still choose to present the full details here, since the proof serves as a warm-up, illustrating several crucial points and providing valuable tools for the (much more difficult) proof of Theorem~\ref{thm-sparse-regime}.

The following polynomials will play a fundamental role in our analysis. 
\begin{DEF}
    For two graphs $S_1,S_2\Subset \operatorname{K}_n$, define the polynomial $\phi_{S_1,S_2}$ associated with $S_1,S_2$ by 
    \begin{equation}\label{eq-def-f-K1K2}
        \phi_{S_1,S_2}\big(\{A_{i,j}\},\{B_{i,j}\}\big)=\big(q(1-q)\big)^{-\frac{|E(S_1)|+|E(S_2)|}{2}}\prod_{(i,j)\in E(S_1)}\overline{A}_{i,j}\prod_{(i,j)\in E(S_2)}\overline{B}_{i,j},
    \end{equation}
    where $\overline{A}_{i,j}=A_{i,j}-q,\overline{B}_{i,j}=B_{i,j}-q$ for all $(i,j)\in \operatorname{U}$. In particular, $\phi_{\emptyset,\emptyset}\equiv 1$.
\end{DEF}
For any \(d \geq 1\), we identify \(\mathcal P_{n,d}\) as a closed subspace of the Hilbert space \(L^2(\{0,1\}^{2|\!\operatorname{U}\!|},\mathbb R)\) with the inner product \(\langle f,g\rangle=\mathbb{E}_{\Qb}[fg]\). It is then straightforward to verify that for any two pairs \((S_1,S_2)\) and \((S_1',S_2')\), we have
\[
\langle \phi_{S_1,S_2},\phi_{S_1',S_2'}\rangle=\mathbb{E}_\Qb[\phi_{S_1,S_2}\phi_{S_1',S_2'}]=\mathbf{1}_{S_1=S_1',S_2=S_2'}.
\]
Therefore, the set of polynomials \(\mathcal O_d=\{\phi_{S_1,S_2}:S_1,S_2\Subset \operatorname{K}_n, |E(S_1)|+|E(S_2)|\leq d\}\) constitutes a standard orthogonal basis for the space \(\mathcal P_{n,d}\). Consequently, we can explicitly express the optimal signal-to-noise ratio of polynomial tests in \(\mathcal P_{n,d}\) in terms of this basis, as in the next lemma.
\begin{lemma}\label{lem-optimal-signal-to-noise-ratio}
    For any $n,d\ge 1$, it holds that
    \begin{equation}\label{eq-optimal-signal-to-noise-ratio}
        \sup_{f\in \mathcal P_{n,d}}\frac{\mathbb E_{\Pb}[f]}{\sqrt{\mathbb E_\Qb [f^2]}}=\left(\sum_{\phi_{S_1,S_2}\in \mathcal O_d}\big(\mathbb E_{\Pb} [\phi_{S_1,S_2}]\big)^2\right)^{1/2}\,.
    \end{equation}
\end{lemma}
\begin{proof}
    For any $f \in \mathcal{P}_{n,d}$, it can be uniquely expressed as $f=\sum_{\phi_{S_1,S_2}\in \mathcal O_d}C_{S_1,S_2} \phi_{S_1,S_2}$ where $C_{S_1,S_2}$'s are real constants. Applying Cauchy-Schwartz inequality one gets
    \begin{align*}
        \frac{ \mathbb{E}_{\mathbb{P}}[f] }{ \sqrt{\mathbb{E}_{\mathbb{Q}}[f^2]} } = \frac{ \sum_{\phi_{S_1,S_2}\in \mathcal O_d} C_{S_1,S_2} \mathbb{E}_{\mathbb{P}}[\phi_{S_1,S_2}] }{ \sqrt{\sum_{\phi_{S_1,S_2}\in \mathcal O_d} C_{S_1,S_2}^2}  } \leq \left(\sum_{\phi_{S_1,S_2}\in \mathcal O_d}\big(\mathbb E_{\Pb} [\phi_{S_1,S_2}]\big)^2\right)^{1/2} \,,
    \end{align*}
    with equality holds if and only if $C_{S_1,S_2}\propto \mathbb{E}_\Pb[\phi_{S_1,S_2}]$. 
\end{proof}

In order to bound the right hand side of \eqref{eq-optimal-signal-to-noise-ratio}, we need the next lemma.
\begin{lemma}\label{lem-expectation-without-conditioning}
    For any two graphs $S_1,S_2\Subset\operatorname{K}_n$, it holds that
    \[
    \mathbb E_\Pb[\phi_{S_1,S_2}]=\begin{cases}
        \rho^{|E(S_1)|}\cdot\frac{\operatorname{Aut}(S_1)}{n(n-1)\cdots(n-|V(S_1)|+1)},\quad&\text{if }S_1\cong S_2,\\
        0\,,&\text{otherwise}.
    \end{cases}
    \]
\end{lemma}
\begin{proof}
Recalling the reparametrization given by $\rho,q$ as in Definition~\ref{def-correlated-random-graph}, for any $(i,j)\in \operatorname{U}$, we have $\mathbb E_{\Pb_*}[A_{i,j}]=\mathbb E_{\Pb_*}[B_{i,j}]=q$ and 
    \[
    \mathbb E_{\Pb_*} \big[(A_{i,j}-q)(B_{\pi_*(i),\pi_*(j)}-q)\big]=\rho q(1-q)\,.
    \]
We observe that conditioned on any realization $\pi$ of $\pi_*$, the pairs of variables $(\overline{A}_{i,j},\overline{B}_{\pi(i),\pi(j)})$ for $(i,j)\in \operatorname{U}$ are mutually independent with $\mathbb E_{\Pb_*[\cdot\mid\pi_*=\pi]}[\overline{A}_{i,j}]=\mathbb{E}_{\Pb_*[\cdot\mid \pi_*=\pi]}[\overline{B}_{\pi(i),\pi(j)}]=0$. As a result, whenever $\pi(K_1)\neq K_2$, we have $\mathbb{E}_{\Pb_*[\cdot \mid\pi_*=\pi]}[\phi_{S_1,S_2}]=0$. 

Clearly, when $S_1 \not \cong S_2$ we have $\pi(S_1) \neq S_2$ holds for any realization $\pi$, and thus $\mathbb{E}_{\mathbb{P}}[\phi_{S_1,S_2}]=0$ by the total probability formula. Otherwise if $S_1 \cong S_2$, then for any $\pi$ such that $\pi(S_1)=S_2$, we have 
$$\mathbb{E}_{\mathbb{P}_*[\cdot\mid\pi_{*}=\pi]}[\phi_{S_1,S_2}]=\big[q(1-q)\big]^{-|E(S_1)|}\prod_{(i,j)\in E(S_1)}\mathbb{E}_{\mathbb{P}_*[\cdot\mid\pi_{*}=\pi]}[\overline{A}_{i,j}\overline{B}_{\pi(i),\pi(j)}]=\rho^{|E(S_1)|}\,.$$ 
Therefore, denoting $\mu_n$ as the uniform measure on permutations of $[n]$, we conclude that
\begin{align*}
    \mathbb{E}_{\mathbb{P}}[\phi_{S_1,S_2}]=&\ \mathbb{E}_{\pi\sim \mu_n} \big[ \mathbb{E}_{\mathbb{P}_*[\cdot \mid\pi_{*}=\pi]}[\phi_{S_1,S_2}] \big]= \rho^{|E(S_1)|} \mu_n[\pi(S_1)=S_2] \\
    =&\ \rho^{|E(S_1)|} \cdot \frac{ \operatorname{Aut}(S_1)(n-|V(S_1)|)!}{n!}=\rho^{|E(S_1)|}\cdot \frac{\operatorname{Aut}(S_1)}{ n(n-1)\ldots(n-|V(S_1)|+1) } \,,
\end{align*}
completing the proof.
\end{proof}
Now we are ready to prove Theorem~\ref{thm-dense-regime}.

\begin{proof}[Proof of Theorem~\ref{thm-dense-regime}]
    It suffices to show that the right hand side of \eqref{eq-optimal-signal-to-noise-ratio} is uniformly bounded whenever $d\le O(\rho^{-1})$. By Lemma~\ref{lem-expectation-without-conditioning}, the problem reduces to controlling
    \begin{align}
    &\ \nonumber\sum_{\mathbf H\in \mathcal H:|E(\mathbf H)|\le d/2}\rho^{2|E(\mathbf H)|}\frac{\operatorname{Aut}(\mathbf H)^2}{[n(n-1)\cdots(n-|V(\mathbf H)|+1)]^2}\cdot\#\big\{S_1,S_2\Subset \operatorname{K}_n:S_1\cong S_2\cong \mathbf H\big\}\\
    =&\ \sum_{\mathbf H\in \mathcal H:|E(\mathbf H)|\le d/2}\rho^{2|E(\mathbf H)|}\,.\label{eq-subgraph-counting-without-conditioning}       \end{align}
    Here the equality follows from the following fact: the number of subgraphs of $\operatorname{K}_n$ that are isomorphic to $\mathbf H$ is equal to $\frac{n(n-1)\cdots(n-|V(\mathbf H)|+1)}{\operatorname{Aut}(\mathbf H)}$. Recalling that $\mathbf H$ has no isolated vertex, we obtain from straightforward computations that \eqref{eq-subgraph-counting-without-conditioning} is upper-bounded by
    \begin{align}
    & \sum_{k\le d/2}\rho^{2k}\sum_{l\le 2k}\#\big\{\mathbf H\in \mathcal H:|E(\mathbf H)|=k,|V(\mathbf H)|=l\big\} \nonumber \\
    \le& \sum_{k\le d/2}\rho^{2k}\sum_{l\le 2k}\binom{l(l-1)/2}{k}\le \sum_{k\le d/2} \rho^{2k}\binom{2k^2}{k+1} \nonumber \\
    \le& \sum_{k\le d/2} \frac{(2\rho k)^{2k}}{(k-1)!}\le \sum_{k\le d/2}\frac{(\rho d)^{2k}}{(k-1)!}\le (\rho d)^2 e^{(\rho d)^2}\,, {\label{eq-display-bounding-(2.5)}}
    \end{align}
    which is $O(1)$ provided that $d=O(\rho^{-1})$, completing the proof.   
\end{proof}

\section{Proof of Theorem~\ref{thm-sparse-regime}}\label{sec-proof-of-sparse}
This section is dedicated to proving Theorem~\ref{thm-sparse-regime}, and we will begin by outlining our proof strategy. First of all, we may assume that \(\rho \geq 1/d\), as otherwise the result can be derived from Theorem~\ref{thm-dense-regime}. As previously suggested, it is crucial to work with a truncated version of \(\Pb\) rather than \(\Pb\) itself. It turns out that an appropriate truncation is to control the density of subgraphs in the mother graph $G$.

In Section~\ref{subsec-admissible}, we introduce the concept of admissible graphs and define \(\Pb_*'\) 
as the conditional probability measure given that all subgraphs of \(G\) are admissible. We further define $\mathbb P'$ to be the marginal of $\mathbb{P}_{*}'$ in the same manner that $\mathbb P$ is the marginal of $\mathbb P_*$. 
Under this conditional measure, the problem of optimization over \(\mathcal P_{n,d}\) can be reduced to optimizing over a much smaller class of polynomials \(\mathcal P_{n,d}'\) (which we refer to as admissible polynomials), as explained in Section~\ref{subsec-reduction-to-P'}. Finally, by analyzing the reduced problem we are able to prove tight hardness results in the sparse regime, as detailed in Section~\ref{subsec-complete-proof}.

Throughout this section, we fix a small constant \(\varepsilon\in (0,0.1)\). Additionally, we fix a sequence \((\rho_n,q_n,d_n)\) with respect to \(n\) such that \(1 \leq nq_n \leq n^{o(1)}\), \(d_n\ge 100\), \(1/d_n^2 \leq \rho_n^2<{\alpha}-\varepsilon\), 
and \(\log d_n\big/\big(\frac{\log n}{\log nq_n}\wedge \sqrt{\log n}\big)\to 0\) as \(n \to \infty\). For the sake of brevity, we will only work with some fixed $n$ throughout this analysis, and we simply denote $q_n,\rho_n,d_n$ as $q,\rho,d$, respectively. While our main interest is to analyze the behavior for sufficiently large $n$, most of our arguments do hold for all $n$, and we will explicitly point out in lemma-statements and proofs when we need the assumption that $n$ is sufficiently large. Several technical results in this section will be proven in the appendix to ensure a smooth flow of presentation.

\subsection{Truncating on admissible graphs} \label{subsec-admissible}

\begin{DEF}\label{def-addmisible}
Given a graph $H=H(V,E)$, define 
    \begin{equation}\label{eq-def-Phi}
        \Phi(H)={\big(n^{1+4/d} {d}^{20}\big)^{|V(H)|} \big(q d^6\big)^{|E(H)|}\,,}
    \end{equation}
and the graph $H$ is said to be \emph{bad} if ${\Phi(H)<(\log n)^{-1}}$. Furthermore, 
    We say a graph is \emph{admissible} if it contains no bad subgraph,  
   and we say it is \emph{inadimissible} otherwise. 
    
    Denote $\Gc$ for the event that $G$ does not contain any bad subgraph with no more than $d^2$ vertices. In addition, let $\Pb_*'$ be the conditional version of $\Pb_*$ given $\Gc$, and let $\Pb'$ be the corresponding marginal distribution of $\Pb_*'$ on $(A,B)$.  
\end{DEF}
\begin{remark}
In the conceptual level, it would be sufficient to truncate on the event that neither \(A\) nor \(B\) contains a bad subgraph with at most \(d\) edges. However, for the sake of convenience in later arguments (specifically, see the proof of Lemma~\ref{lemma-prob-proper-realization} in Section~\ref{subsec-B2} for details), we opt to further truncate based on the condition that \(G\) does not contain  a bad subgraph with at most \(d^2\) vertices. It is clear that $\mathcal G$ inplies that $G$ contains no inadmissible subgraph with at most $d^2$ vertices. In addition, it is worthwhile noting that any subgraph of an admissible graph is still admissible.
\end{remark}
\begin{lemma}\label{eq-Gc-is-typical}
    For any permutation $\pi\in \operatorname{S}_n$, it holds that $\Pb_*[\Gc\mid \pi_*=\pi]=1-o(1)$. Therefore, $\Pb_*[\Gc]=1-o(1)$ and $\operatorname{TV}(\Pb,\Pb')\le \operatorname{TV}(\Pb_*,\Pb_*')=o(1)$.
\end{lemma}
\begin{proof}
Note that when $\rho \geq 1/d$, $G$ is an \ER graph with edge density $p \leq qd$, and $G$ is independent with $\pi_*$. Hence it suffices to show that with vanishing probability such an \ER graph contains a bad subgraph of size at most $d^2$. We shall do this via a union bound.

For $W\subset [n]$ define $G_{W}=(W,E_{W})$ to be the induced subgraph of $G$ in $W$. 
For $1 \leq k \leq d^2$, define
\begin{equation} {\label{eq-def-E(k)}}
    e(k) = \min \big\{ k' \geq 0 : \big(n^{1+4/d} {d}^{20}\big)^{k} \big(q d^6\big)^{k'} < (\log n)^{-1} \big\}  \,.
\end{equation}
Since $\Gc^c$ is equivalent to that there exists $W \subset [n]$ with $|W|=k \leq d^2$ and $|E_W|\ge e(k)$, we conclude from a union bound that 
\begin{align}{ \label{eq-upper-bound-prob-G^c} }
    \Pb_*[\Gc^c] \le \sum_{k=1}^{d^2} \sum_{W\subset[n],|W|=k} {\mathbb{P}_*}\Big[ |E_W| \geq e(k) \Big] = \sum_{k=1}^{{d^2}}\binom{n}{k} \mathbb{P} \Big[ \mathbf{B}\Big(\binom{k}{2},p\Big) \geq e(k) \Big]\,,
\end{align}
where $\mathbf B\big(\binom{k}{2},q\big)$ is a binomial variable with parameters $\big(\binom{k}{2},q\big)$, and it has the same distribution as $|E_W|$ for any $W\subset [n]$ with $|W|=k$. Since $nq\le n^{o(1)}$ and $d\ge 100$, 
it is easy to verify that $k\le e(k)\le 2k$. For $k\le d^2=n^{o(1)}$, we have $\binom{k}{2}p=o(1)$ and thus by Poisson approximation we get (recalling that we have assumed $\rho \ge 1/d$ and thus $p \le qd$) 
\begin{equation*}
    \binom{n}{k} \mathbb{P} \Big[ \mathbf{B} \Big(\binom{k}{2},p\Big) \geq e(k) \Big] \leq n^k (k^2 p)^{ e(k) } \leq \big({d}^{20}n^{1+4/d}\big)^k \big({d}^6 q\big)^{e(k)} \cdot \big( d^{20} n^{4/d} \big)^{-k} \,,
\end{equation*}
which is upper-bounded by ${d}^{-20k} n^{-4k/d} (\log n)^{-1}$ from the definition of $e(k)$ in \eqref{eq-def-E(k)}. Plugging this estimation into \eqref{eq-upper-bound-prob-G^c}, we get that
\[
    \Pb_*[\Gc^c]\le (\log n)^{-1} \sum_{k=1}^{{d^2}} n^{-4k/d} d^{-20k}=o(1) \,,
\]
which shows that $\Pb_*[\Gc\mid \pi_*=\pi]=\Pb_*[\Gc]=1-o(1)$ for any $\pi\in\operatorname{S}_n$. The statement on the total variational distance in the lemma then follows from the data processing inequality.
\end{proof}

We will show in the remaining of this section that $\eqref{eq-evidence-detection}'$ is indeed true for $\Pb'$ and $\Pb_*'$ defined as above.

\subsection{Reduction to admissible polynomials}\label{subsec-reduction-to-P'}
Recall that \(\mathcal P_{n,d}\) represents the set of real polynomials on \(\{0,1\}^{2|\!\operatorname{U}\!|}\) with degree no more than $d$, and we have shown that \(\mathcal O_d=\{\phi_{S_1,S_2}:S_1,S_2\Subset \operatorname{K}_n, |E(S_1)|+|E(S_2)| \leq d\}\) forms a basis for \(\mathcal P_{n,d}\). Now we say a polynomial \(\phi_{S_1,S_2}\in \mathcal O_d\) is \emph{admissible} if both \(S_1\) and \(S_2\) are admissible graphs. Furthermore, we define $\mathcal O_d'\subset \mathcal O_d$ as the set of admissible polynomials in $\mathcal O_d$, and define \(\mathcal P_{n,d}' \subset \mathcal P_{n,d}\) as the linear subspace spanned by polynomials in $\mathcal O_d'$. 

Intuitively, due to the absence of inadmissible graphs under the law \(\Pb'\), only admissible polynomials are relevant in polynomial-based algorithms. Therefore, it is plausible to establish $\eqref{eq-evidence-detection}'$ by restricting to polynomials in $\mathcal P_{n,d}'$. The following proposition as well as the discussion afterwards formalize this intuition.

\begin{proposition}\label{prop-same-L^1-bounded-L^2}
    For any $f\in \mathcal P_{n,d}$, there exists some $f'\in \mathcal P_{n,d}'$ such that $\mathbb E_{\Qb}[(f')^2]\le 8\mathbb{E}_{\Qb}[f^2]$ and $f'=f$ a.s. under both $\Pb_*'$ and $\Pb'$. 
\end{proposition}
Provided with Proposition \ref{prop-same-L^1-bounded-L^2},
we immediately get that
    \begin{equation}\label{eq-reduction-1}
        \sup_{f\in \mathcal P_{n,d}}\frac{\mathbb E_{\Pb'}[f]}{\sqrt{\mathbb E_{\Qb}[f^2]}}\le 2\sqrt 2\sup_{f\in \mathcal P_{n,d}'}\frac{\mathbb E_{\Pb'}[f]}{\sqrt{\mathbb E_{\Qb}[f^2]}}\,.
    \end{equation}
    Thus, we successfully reduce the optimization problem over $\mathcal P_{n,d}$ to that over $\mathcal P_{n,d}'$ (up to a multiplicative constant factor, which is not material).

    Now we turn to the proof of Proposition~\ref{prop-same-L^1-bounded-L^2}. 
    For variables $X\in \{A,B\}$ (meaning that $X_{i,j}=A_{i,j}$ or $X_{i,j}=B_{i,j}$ for all $(i,j)\in \operatorname{U}$), denote for $S\Subset \operatorname{K}_n$ that 
    \begin{equation}\label{eq-def-f_S}
    \psi_{S}(\{X_{i,j}\}_{(i,j)\in \operatorname{U}})
    =\prod_{(i,j)\in E(S)} \frac{(X_{i,j}-q)}{\sqrt{q(1-q)}}\,.
    \end{equation}
   Recalling the definition of $\phi_{S_1, S_2}$, we can write it as follows:
    $$\phi_{S_1,S_2}(A,B)=\prod_{(i,j)\in E(S_1)}\frac{(A_{i,j}-q)}{\sqrt{q(1-q)}}\prod_{(i,j)\in E(S_2)}\frac{(B_{i,j}-q)}{\sqrt{q(1-q)}}=\psi_{S_1}(A)\psi_{S_2}(B)\,.$$
   In light of this, we next analyze the polynomial $\psi_S(X)$ via the following expansion:
    \begin{equation*}
\psi_S(X)=\sum_{K\subset S}\left(-\frac{\sqrt q}{\sqrt{1-q}}\right)^{|E(S)|-|E(K)|}\prod_{(i,j)\in E(K)}\frac{X_{i,j}}{\sqrt{q(1-q)}}\,,
    \end{equation*}
    where the summation is taken over all subgraphs of $S$ without isolated vertices (there are $2^{|E(S)|}$ many of them). We define the ``inadmissible-part-removed'' version of $\psi_S(X)$ by 
    \begin{equation}\label{eq-def-f_S'}
    \Hat{\psi}_{S}(X)=\sum_{\substack {K\subset S\\K\text{ is admissible}}}\left(-\frac{\sqrt q}{\sqrt{1-q}}\right)^{|E(S)|-|E(K)|}\prod_{(i,j)\in E(K)}\frac{X_{i,j}}{\sqrt{q(1-q)}}\,,
    \end{equation}
   and obviously we have that $\psi_S(A)-\Hat{\psi}_S(A)=\psi_S(B)-\Hat{\psi}_S(B)=0$ a.s. under both $\Pb_*'$ and $\Pb'$. The reduction here is quite natural, and we next further introduce some quantities in order to facilitate the proof of Proposition~\ref{prop-same-L^1-bounded-L^2}. 
\begin{DEF}
    
For $S\Subset \operatorname{K}_n$, define 
\begin{equation}{\label{eq-def-psi-S}} 
    D(S)=
    \begin{cases}
        \emptyset, & \text{if }S \mbox{ is admissible}\,, \\
        \arg \min_{H \subset S} \{ \Phi(H) \}, &\text{if } S \mbox{ is inadmissible}\,,
    \end{cases}
\end{equation}
(if there are multiple choices of $D(S)$ we arbitrarily take one of them)
and define 
\begin{equation}\label{eq-def-A(S)}
    \mathcal A(S)=\{H\subset S:H\setminus D(S)=S\setminus D(S), H\cap D(S)\text{ is admissible}\}\,.
\end{equation}
We also define the polynomial (recall \eqref{eq-def-f_S} and \eqref{eq-def-f_S'})
\begin{equation}\label{eq-def-hat-f_S}
    \psi_S'(\{X_{i,j}\}_{(i,j)\in \operatorname{U}})=\psi_{S\setminus D(S)}(\{X_{i,j}\}_{(i,j)\in \operatorname{U}})\cdot \Hat{\psi}_{D(S)}(\{X_{i,j}\}_{(i,j)\in \operatorname{U}})\,.
\end{equation}
Moreover, we define 
\begin{equation}\label{eq-def-hat-f-S1S2}
    \phi_{S_1,S_2}'(A,B)=\psi_{S_1}'(A)\psi_{S_2}'(B)\,, \forall S_1,S_2\subset K_n\,.
\end{equation}
Then it holds that $\phi_{S_1,S_2}'(A,B)=\phi_{S_1,S_2}(A,B)$ a.s. under both $\Pb_*'$ and $\Pb'$.
\end{DEF}
\begin{lemma}\label{lem-psi(S)-property}
    For any inadmissible graph $S\Subset \operatorname{K}_n$ and any $H\in \mathcal A(S)$, it holds that $H$ itself is admissible and $\Phi(H)\ge \Phi(S)$. Furthermore, any $\psi_S'$ is a linear combination of $\{ \psi_H:H\in \mathcal{A}(S)\}$. As a result, $\phi_{S_1,S_2}'\in \mathcal P_{n,d}'$ for any $S_1,S_2\Subset \operatorname{K}_n\text{ with }|E(S_1)|+|E(S_2)|\le d$.
\end{lemma}
\begin{proof}
    Take any inadmissible graph $S\Subset \operatorname{K}_n$ and $H\in \mathcal A(S)$. For any subgraph $H'\subset H\subset S$, from the definition of $D(S)$ in \eqref{eq-def-psi-S} we get that 
\begin{align*}
    \Phi(D(S)) \leq \Phi(D(S) \cup H') \overset{ \operatorname{Lemma}~\ref{lemma-facts-graphs} \textup{(ii)} }{\le} \frac{ \Phi(D(S)) \Phi(H') }{ \Phi(H' \cap D(S)) } \,,
\end{align*}
and thus we have (using $H\cap D(S)$ is admissible and $H'\cap D(S) \subset H\cap D(S)$)
\begin{align*}
    \Phi(H') \geq \Phi(H' \cap D(S)) \geq (\log n)^{-1} \,.
\end{align*}
This verifies that $H$ is admissible. Furthermore, 
\[
\Phi(D(S)\cap H)\ge\Phi(D(S))\overset{ \operatorname{Lemma}~\ref{lemma-facts-graphs} \textup{(ii)} }{\ge}\frac{\Phi(D(S)\cup H)\Phi(D(S)\cap H)}{\Phi(H)}=\frac{\Phi(S)\Phi(D(S)\cap H)}{\Phi(H)}\,,
\]
where the first inequality follows from the minimality of $\Phi(D(S))$ and the equality follows from the fact that $H \cup D(S)=S$ for $H\in \mathcal A(S)$. This reduces to $\Phi(H)\ge \Phi(S)$ and thus concludes the first statement.

Now for $X\in \{A,B\}$, recall from \eqref{eq-def-f_S'} that $\Hat{\psi}_{D(S)}(X)$ is a linear combination of polynomials in the set 
$$\Big\{\prod_{(i,j)\in E(K)}X_{i,j}: K\subset D(S), K\text{ is admissible}\Big\}\,.$$ 
Note that $\{\psi_S(X):S\Subset \operatorname{K}_n\text{ with }|E(S)|\le d\}$ is an orthogonal basis for the space of polynomials of $X\in \{0,1\}^{\operatorname{U}}$ with degree no more than $d$. In addition, we claim that for $T\Subset\operatorname{K}_n$, the inner product
$$\langle \Hat{\psi}_{D(S)}(X),\psi_{T}(X) \rangle=\mathbb{E}_{\Qb}[\Hat{\psi}_{D(S)}(X) \psi_{T}(X)] \neq 0$$ 
only if $T \subset D(S)$ and $T$ is admissible. To see this, from the definition of $\Hat{\psi}_{D(S)}$ as in \eqref{eq-def-f_S'} we get that for any $T$ such that the inner product is non-zero, there exists an admissible subgraph $K$ of $D(S)$ satisfying $\langle \prod_{(i,j) \in E(K)} X_{i,j}, \psi_T(X) \rangle \neq 0$. This can only happen when $T \subset K$, since otherwise one can extract a factor $\mathbb{E}[X_{e}-q]=0$ with $e\in E(T)\setminus E(K)$ to obtain a contradiction. Thus, we have $T \subset K \subset D(S)$ and $T$ is admissible. Therefore, $\Hat{\psi}_{D(S)}(X)$ is also a linear combination of polynomials in $\{\psi_{K}(X): K\subset D(S)\text{ is admissible}\}$. This shows that $\psi'_S=\psi_{S\setminus D(S)}\cdot \Hat{\psi}_{D(S)}$ is a linear combination of $\{\psi_H:H\in \mathcal A(S)\}$ by \eqref{eq-def-A(S)}, which concludes the second statement in the lemma.

Combining what we proved together, we get that for any inadmissible $S\Subset\operatorname{K}_n$, $\psi_S'(X)$ is a linear combination of polynomials $\psi_K(X)$, where $K$ ranges over admissible subgraph of $S$. Clearly the claim is also true for any admissible $S\Subset \operatorname{K}_n$ as by definition $\psi_S'(X)=\psi_S(X)$. 
Hence we conclude that $\phi_{S_1,S_2}'(A,B)=\psi_{S_1}'(A)\psi_{S_2}'(B)$ is a linear combination of $\psi_{K_1}(A)\psi_{K_2}(B)$ with $K_1,K_2\Subset \operatorname{K}_n$ that are admissible. This shows that $\phi_{S_1,S_2}'\in \mathcal P_{n,d}'$ for any $S_1,S_2\Subset \operatorname{K}_n$ with $|E(S_1)|+|E(S_2)|\le d$, completing the proof.
\end{proof}
We now elaborate on the polynomials $\psi_S'(X)$ more carefully. Write
\begin{equation}
\label{eq-def-hatpsi-Lambda-H,S}    
\psi_S'(X)=\sum_{H\in \mathcal A(S)}\Lambda_S(H) \psi_{H}(X)\,,
\end{equation}
and we turn to determine the coefficients $\Lambda_S(H)$ explicitly. From the orthogonal property we get
\begin{equation}\label{eq-Lambda(H,S)}
    \Lambda_S(H)=\big\langle \psi_{S}'(X), \psi_{H}(X) \big\rangle =\mathbb{E}_{\mathbb{Q}} \big[\psi_{S}'(X)\psi_{H}(X) \big] \,.  
\end{equation}
From the definition~\eqref{eq-def-f_S'} and \eqref{eq-def-hat-f_S}, we see \eqref{eq-Lambda(H,S)} equals to
\begin{align*}
    &\ \ \sum_{ K \subset D(S)\text{ admissible} } \Big(- \frac{\sqrt{q}}{\sqrt{1-q}} \Big)^{ |E(D(S))|-|E(K)| } \Big \langle \psi_{S\setminus D(S)}(X) \prod_{ (i,j) \in E(K)} \frac{X_{i,j}}{\sqrt{q(1-q)}}, \psi_{H}(X) \Big \rangle\\
    &\qquad=\sum_{J\subset \mathcal A(S)}\Big(- \frac{\sqrt{q}}{\sqrt{1-q}} \Big)^{ |E(S)|-|E(J)| } \Big \langle \psi_{S\setminus D(S)}(X) \prod_{ (i,j) \in E(J\cap D(S))} \frac{X_{i,j}}{\sqrt{q(1-q)}}, \psi_{H}(X) \Big \rangle\,.
\end{align*}
Here the equality follows by identifying each admissible subgraph $K$ of $D(S)$ with $J=K\cup (S\setminus D(S))\in \mathcal A(S)$ and noting that $|E(D(S))|-|E(K)|=|E(S)|-|E(J)|$.
Since $H \setminus D(S)=J \setminus D(S)=S \setminus D(S)$ for $H,J\in\mathcal{A}(S)$, a simple calculation yields 
\begin{align*}
    &\ \Big \langle \psi_{S\setminus D(S)}(X) \prod_{ (i,j) \in E(J \cap D(S)) } \frac{X_{i,j}}{\sqrt{q(1-q)}}, \psi_{H}(X) \Big \rangle  \\
    =&\ \Big \langle \psi_{S\setminus D(S)}(X)\prod_{ (i,j) \in E(J \cap D(S)) } \frac{X_{i,j}}{\sqrt{q(1-q)}}, \psi_{S\setminus D(S)}(X) \cdot \psi_{H \cap D(S)}(X) \Big \rangle \\
    = &\ \mathbb{E}_{\mathbb{Q}} \Big[ \prod_{(i,j) \in E(J\cap D(S)) } \frac{X_{i,j}}{\sqrt{q(1-q)}} \prod_{(i',j') \in E(H\cap D(S))} \frac{X_{i',j'}-q}{\sqrt{q(1-q)}} \Big] = \mathbf{1}_{ \{H \subset J\} } \Big( \frac{ \sqrt{q} }{ \sqrt{1-q} } \Big)^{ |E(J)|-|E(H)| } \,.
\end{align*}
Therefore, we conclude that
\begin{equation}{\label{eq-def-Lambda-H,S}}
    \Lambda_S(H) = \Big( \frac{ \sqrt{q} }{ \sqrt{1-q} } \Big)^{|E(S)|-|E(H)|} \sum_{J \in \mathcal{A}(S) , H \subset J \subset S} (-1)^{|E(S)|-|E(J)|}  \,,  
\end{equation} 
and hence that
\begin{equation}\label{eq-est-of-Lambda}
    \begin{aligned}
    &\  |\Lambda_S(H)| \leq \Big( \frac{ \sqrt{q} }{ \sqrt{1-q} } \Big)^{ |E(S)|-|E(H)| } \cdot \# \big\{ J \in \mathcal{A}(S): H \subset J \subset S \big\} \\
    \leq &\  (2\sqrt{q})^{ |E(S)|-|E(H)| } \cdot \# \big\{ J : H \subset J \subset S \big\} \leq (4\sqrt{q})^{ |E(S)|-|E(H)| } \,.
    \end{aligned}
\end{equation}

With these estimates in hand, we are now ready to prove Proposition~\ref{prop-same-L^1-bounded-L^2}. 
\begin{proof}[Proof of Proposition~\ref{prop-same-L^1-bounded-L^2}]
For any $f\in \mathcal P_{n,d}$, we write $f=\sum_{\phi_{S_1,S_2}\in \mathcal O_d}C_{S_1,S_2}\phi_{S_1,S_2}$ as before, and define $ f'=\sum_{\phi_{S_1,S_2}\in \mathcal O_d} C_{S_1,S_2} \phi_{S_1,S_2}'$. Then it is clear that $f'(A,B)=f(A,B)$ a.s. under $\Pb_*'$, and from Lemma~\ref{lem-psi(S)-property} that $f'\in \mathcal P_{n,d}'$. Now we show that $\mathbb E_\Qb [(f')^2]\le 8\mathbb E_\Qb [f^2]$.

Recalling \eqref{eq-def-hat-f-S1S2} and \eqref{eq-def-hatpsi-Lambda-H,S}, we have that
\begin{align}
    \phi_{S_1,S_2}'(A,B) &= \Big( \sum_{H_1 \in \mathcal{A}(S_1)} \Lambda_{S_1}(H_1) \psi_{H_1}(A) \Big) \cdot \Big( \sum_{H_2 \in \mathcal{A}(S_2)} \Lambda_{S_2}(H_2) \psi_{H_2}(B) \Big) \nonumber \\
    &= \sum_{  H_1 \in \mathcal{A}(S_1), H_2 \in \mathcal{A}(S_2) } \Lambda_{S_1}(H_1) \Lambda_{S_2}(H_2) \phi_{H_1,H_2}(A,B) \,. \label{equ-simple-form-phi'}
\end{align}
Thus, $f'$ can also be written as
\begin{align*}
    f'&=\sum_{(S_1,S_2)\Subset \operatorname{K}_n:|E(S_1)|+|E(S_2)|\le d}C_{S_1,S_2}\Big(\sum_{H_1\in \mathcal A(S_1),H_2\in\mathcal A(S_2)} \Lambda_{S_1}(H_1)\Lambda_{S_2}(H_2)\phi_{H_1,H_2}(A,B)\Big)\\
    &= \sum_{ H_1,H_2 \operatorname{admissible} }\Big( \sum_{\substack{(S_1,S_2)\Subset \operatorname{K}_n:|E(S_1)|+|E(S_2)|\le d\\  H_1 \in \mathcal{A}(S_1), H_2 \in \mathcal{A}(S_2)} } C_{S_1,S_2} \Lambda_{S_1}(H_1) \Lambda_{S_2}(H_2) \Big) \phi_{H_1,H_2}(A,B) \,.
\end{align*}
Therefore, by the orthogonality of $\{\phi_{S_1,S_2}\}$, we have that $\mathbb E_\Qb[(f')^2]$ is upper-bounded by
\begin{align}
    \nonumber& \sum_{ H_1,H_2 \operatorname{admissible} } \Big( \sum_{ \substack{(S_1,S_2)\Subset\operatorname{K}_n:|E(S_1)|+|E(S_2)|\le d\\  H_1 \in \mathcal{A}(S_1), H_2 \in \mathcal{A}(S_2)}} C_{S_1,S_2} \Lambda_{S_1}(H_1) \Lambda_{S_2}(H_2) \Big)^2 
    \\
    \nonumber\stackrel{\eqref{eq-est-of-Lambda}}{\leq} & \sum_{H_1,H_2\text{ admissible}} \Big( \sum_{ \substack{(S_1,S_2)\Subset\operatorname{K}_n:|E(S_1)|+|E(S_2)|\le d\\  H_1 \in \mathcal{A}(S_1), H_2 \in \mathcal{A}(S_2)} } (16q)^{ \frac{1}{2}(|E(S_1)|+|E(S_2)|-|E(H_1)|-|E(H_2)|) } |C_{S_1,S_2}| \Big)^2 \\
    \nonumber\leq & \sum_{ H_1,H_2 \operatorname{admissible} } \Big( \sum_{ \substack{(S_1,S_2)\Subset\operatorname{K}_n:|E(S_1)|+|E(S_2)|\le d\\  H_1 \in \mathcal{A}(S_1), H_2 \in \mathcal{A}(S_2)} } {d}^{ -2 (|E(S_1)|+|E(S_2)|-|E(H_1)|-|E(H_2)|) } C^2_{S_1,S_2} \Big) \\
    \label{eq-last-bracket}& \times \Big( \sum_{ \substack{(S_1,S_2)\Subset\operatorname{K}_n:|E(S_1)|+|E(S_2)|\le d\\  H_1 \in \mathcal{A}(S_1), H_2 \in \mathcal{A}(S_2)}} (16q {d}^2)^{(|E(S_1)|+|E(S_2)|-|E(H_1)|-|E(H_2)|)} \Big)\,,
\end{align}
where the last inequality follows from Cauchy-Schwartz inequality.

We next upper-bound the right hand side of \eqref{eq-last-bracket}. To this end, we first show that the last bracket in \eqref{eq-last-bracket} is uniformly bounded by $4$ for any two admissible graphs $H_1,H_2$. 
Note that
\begin{align}\label{eq-bound-second-bracket}
    & \sum_{ S_1\Subset\operatorname{K}_n:H_1 \in \mathcal{A}(S_1), |E(S_1)|\leq d } (16q{d}^2)^{ (|E(S_1)|-|E(H_1)|) } \nonumber \\ 
    \leq & \sum_{v,k=0}^{2d} 16^{k}q^{k}{d}^{2k} \# \big\{ S_1 : H_1 \in \mathcal{A}(S_1), |V(S_1)|-|V(H_1)|=v,|E(S_1)|-|E(H_1)|=k \big\} \nonumber \\
    \leq & \sum_{v,k=0}^{2d} 16^{k}q^{k}n^{v} {d}^{4k} \mathbf{1}_{ \{ \exists S_1: H_1 \in \mathcal{A}(S_1), |V(S_1)|-|V(H_1)|=v,|E(S_1)|-|E(H_1)|=k \} }  \,,
\end{align}
where the last step comes from Lemma~\ref{lemma-facts-graphs} (iv). By Lemma~\ref{lem-psi(S)-property} we have $\Phi(H_1) \geq \Phi(S_1)$ for any $H_1 \in \mathcal{A}(S_1)$, which implies $(n^{1+4/d}d^{20})^{v} (q d^6)^k \leq 1$ whenever the indicator equals to $1$. Thus,
\begin{align*}
    \eqref{eq-bound-second-bracket} \leq \sum_{v,k=0}^{2d} 16^k {d}^{4k-20v-6k} \leq \sum_{v,k=0}^{2d} {d}^{-(v+k)} \leq 2.
\end{align*}
We can get another inequality with respect to $H_2$ in the same way and this verifies our claim, since the last bracket in \eqref{eq-last-bracket} is upper-bounded by the product of these two sums. Therefore, we get that $\mathbb{E}_{\mathbb{Q}}[(f')^2]$ is upper-bounded by $4$ times
\begin{align*}
     & \sum_{ H_1,H_2 \operatorname{admissible} } \Big( \sum_{ \substack{(S_1,S_2)\Subset\operatorname{K}_n:|E(S_1)|+|E(S_2)|\le d\\  H_1 \in \mathcal{A}(S_1), H_2 \in \mathcal{A}(S_2)} } {d}^{ -2 (|E(S_1)|+|E(S_2)|-|E(H_1)|-|E(H_2)|) } C^2_{S_1,S_2} \Big) \\
     =&\sum_{(S_1,S_2)\Subset\operatorname{K}_n:|E(S_1)|+|E(S_2)|\le d} C^2_{S_1,S_2} \sum_{ H_1 \in \mathcal{A}(S_1), H_2 \in \mathcal{A}(S_2) } {d}^{ -2(|E(S_1)|+|E(S_2)|-|E(H_1)|-|E(H_2)|) } \,.
\end{align*}
In addition, for any fixed $S_1,S_2$ such that $|E(S_1)|\le d$ and $|E(S_2)| \leq d$, we have that
\begin{align*}
    &\ \sum_{ H_1 \in \mathcal{A}(S_1), H_2 \in \mathcal{A}(S_2) } {d}^{-2 (|E(S_1)|+|E(S_2)|-|E(H_1)|-|E(H_2)|)} \\\le &\ \sum_{k_1=0}^{|E(S_1)|}\sum_{k_2=0}^{|E(S_2)|}d^{-2(k_1+k_2)}\cdot \#\{(H_1,H_2):H_i\subset S_i:|E(H_i)|=|E(S_i)|-k_i,i=1,2\}
    \\
    \le&\ 
    \sum_{k_1=0}^{|E(S_1)}\sum_{k_2=0}^{|E(S_2)|}d^{-2(k_1+k_2)}\cdot|E(S_1)|^{k_1}|E(S_2)|^{k_2}\le \sum_{k_1=0}^d\sum_{k_2=0}^d d^{-2(k_1+k_2)}\cdot d^{k_1+k_2}\leq 2 \,,
\end{align*}
where the second inequality follows from Lemma~\ref{lemma-facts-graphs} (v) and the last one comes from the fact that $d\ge 100$. Hence, we have $\mathbb{E}_{\mathbb{Q}}[(f')^2] \leq 8 \sum_{S_1,S_2}C_{S_1,S_2}^2 = 8\mathbb{E}_{\mathbb{Q}}[f^2]$, completing the proof of Proposition~\ref{prop-same-L^1-bounded-L^2}.
\end{proof}

\subsection{Completing the proof}\label{subsec-complete-proof}
Now it remains to bound the right hand side of \eqref{eq-reduction-1}. Recall that $$\mathcal O_d'=\{\phi_{S_1,S_2}\in \mathcal O_d:S_1,S_2\text{ are admissible}\}$$
is a standard orthogonal basis of $\mathcal P_{n,d}'$. The following result provides an analogue of Lemma~\ref{lem-optimal-signal-to-noise-ratio}.
\begin{lemma}\label{lem-optimal-results-conditioning}
    For any $n,d\ge 1$, it holds that
    \begin{equation}\label{eq-optimal-signal-to-noise-ratio-conditioning}
        \sup_{f\in \mathcal P_{n,d}'}\frac{\mathbb E_{\Pb'}[f]}{\sqrt{\mathbb E_{\Qb}[f^2]}}=\Bigg(\sum_{\phi_{S_1,S_2}\in \mathcal O_d'} \big(\mathbb E_{\Pb'}[\phi_{S_1,S_2}]\big)^2 \Bigg)^{1/2}\,.
    \end{equation}
\end{lemma}
The proof of Lemma~\ref{lem-optimal-results-conditioning} follows from an almost identical argument to that provided for proving Lemma~\ref{lem-optimal-signal-to-noise-ratio}, so we omit the details here. The rest of this section is dedicated to upper-bounding the right-hand sides of \eqref{eq-optimal-signal-to-noise-ratio-conditioning}. 

A new challenge in bounding \eqref{eq-optimal-signal-to-noise-ratio-conditioning} is that there is no straightforward expression for the expectation of \(\phi_{S_1,S_2}\) under \(\Pb'\). To tackle this, we prove the following bound as an alternative.
\begin{proposition}\label{prop-controlling-expectation}
    For any two admissible graphs $S_1,S_2\Subset \operatorname{K}_n$ with at most $d$ edges, one has that $\big|\mathbb E_{\Pb'}[\phi_{S_1,S_2}]\big|$ is upper-bounded by $1+o(1)$ times
    \begin{equation}\label{eq-control-expectation}
    \begin{aligned}
    \sum_{\mathbf H_0\in \mathcal H:\mathbf H_0\xhookrightarrow{} S_i,i=1,2}n^{-\frac{|V(S_1)|+|V(S_2)|}{2}} \rho^{|E(\mathbf H_0)|} d^{-6( |E(S_1)|+|E(S_2)|-2|E(\mathbf H_0)| )} \operatorname{Aut}(\mathbf H_0) \,,
    \end{aligned}
    \end{equation}
    where the notation $\mathbf H\xhookrightarrow{} S$ means that $\mathbf H$ can be embedded into $S$ as a subgraph.  
\end{proposition}
For a deterministic permutation $\pi$, we use \(\Pb_\pi\) and \(\Pb'_\pi\) to represent \(\Pb_*[\cdot\mid \pi_*=\pi]\) and \(\Pb_*'[\cdot\mid \pi_*=\pi]\), respectively. Additionally, we denote \(\mathbb{E}_\pi\) (resp. \(\mathbb{E}'_\pi\)) as the expectation under \(\Pb_\pi\) (resp. \(\Pb'_\pi\)). It is clear that
\begin{align}
    &\ \big|\mathbb{E}_{\Pb'}[\phi_{S_1,S_2}]\big|= \left|\frac{1}{n!}\sum_{\pi\in \operatorname{S}_n}\mathbb{E}_{\Pb'_\pi}[\phi_{S_1,S_2}]\right| \le \frac{1}{n!}\sum_{\pi\in\operatorname{S}_n}\big|\mathbb{E}_{\Pb_\pi'}[\phi_{S_1,S_2}]\big| \,. \label{eq-relaxation-1} 
\end{align}
We proceed to estimate each term \(|\mathbb{E}_{\Pb_\pi'}[\phi_{S_1,S_2}]|\) in two ways. The first approach is to bound this term simply by the \(L^1\)-norm under \(\Pb_\pi\), as incorporated in the next lemma.
\begin{lemma}\label{lem-L^1-bound}
    For any $S_1,S_2\Subset \operatorname{K}_n$ with at most $d$ edges and any permutation $\pi$ on $[n]$, we have that (denoting $S_0=S_1\cap \pi^{-1}(S_2)\Subset \operatorname{K}_n$ below)
    \begin{equation}\label{eq-L^1-bound}
        \big|\mathbb E_{\Pb_\pi'}[\phi_{S_1,S_2}]\big|\le \big[1+o(1)\big]\rho^{|E(S_0)|}(4q)^{\frac{|E(S_1)|+|E(S_2)|}{2}-|E(S_0)|}\,.
    \end{equation}
\end{lemma}
\begin{proof}
Denote $S_0'=\pi(S_0)$ for short. From the definition of the model, we see that under the measure $\Pb_\pi$, the families of variables
\[
\{A_{i,j},B_{\pi(i),\pi(j)}\}_{(i,j)\in E(S_0)},\{A_{i,j}\}_{(i,j)\in E(S_1)\setminus E(S_0)},\{B_{i,j}\}_{(i,j)\in E(S_2)\setminus E(S_0')}
\]
are mutually independent. As a result one has
\begin{align}
    &\ \big|\mathbb E_{\Pb'_\pi}[\phi_{S_1,S_2}]\big|\le \mathbb E_{\Pb_\pi'}\big[|\phi_{S_1,S_2}|\big]
    \le \frac{1}{\Pb'_\pi[\Gc]}\cdot\mathbb{E}_{\Pb_\pi}\big[|\phi_{S_1,S_2}|\big] \nonumber \\
    =&\ \big[1+o(1)\big]{\big(q(1-q)\big)^{-\frac{|E(S_1)|+|E(S_2)|}{2}}}\cdot\mathbb{E}_{\Pb_\pi}\Big[\prod_{(i,j)\in E(S_1)}|\overline{A}_{i,j}|\prod_{(i,j)\in E(S_2)}|\overline{B}_{i,j}|\Big]\,, \label{eq-first-moment-step-1}
\end{align}
In addition, we have that 
\begin{align}
    &\mathbb{E}_{\Pb_\pi}\prod_{(i,j)\in E(S_1)}|\overline{A}_{i,j}|\prod_{(i,j)\in E(S_2)}|\overline{B}_{i,j}| \nonumber \\
    =& \prod_{(i,j)\in E(S_0)}\mathbb{E}_{\Pb_\pi}|\overline{A}_{i,j}\overline{B}_{\pi(i),\pi(j)}| \prod_{(i,j)\in E(S_1)\setminus E(S_0)}\mathbb{E}_{\Pb_\pi}|\overline{A}_{i,j}|\prod_{(i,j)\in E(S_2)\setminus E(S_0')}\mathbb{E}_{\Pb_{\pi}}|\overline{B}_{i,j}| \nonumber \\
    =& \big[ q(1-q)\big(\rho(1-2q)^2+4q(1-q)\big) \big]^{|E(S_0)|} \cdot \big[2q(1-q)\big]^{|E(S_1)|+|E(S_2)| -2|E(S_0)|} \,, \label{eq-first-moment-step-2}
\end{align}
where the last equality follows from the fact that $\rho = \frac{s(1-q)}{1-qs}$ as in Definition~\ref{def-correlated-random-graph}.
Note that for $|E(S_0)|\le d=n^{o(1)}$ and $\rho\ge 1/d$, it holds that 
\[
\big[ q(1-q)\big(\rho(1-2q)^2+4q(1-q)\big) \big]^{|E(S_0)|}= \big[1+o(1)\big]\rho^{|E(S_0)|}q^{|E(S_0)|} \,.
\]
Combined with \eqref{eq-first-moment-step-1} and \eqref{eq-first-moment-step-2}, it yields that 
\[
\big|\mathbb E_{\Pb'_\pi}[\phi_{S_1,S_2}]\big|\le \big[1+o(1)\big]\rho^{|E(S_0)|}\big(4q\big)^{\frac{|E(S_1)|+|E(S_2)|}{2}-|E(S_0)|}\,.\qedhere
\]
\end{proof}
In addition to Lemma~\ref{lem-L^1-bound}, the following proposition provides another bound on the quantity $\big|\mathbb E_{\Pb_\pi'}[\phi_{S_1,S_2}]\big|$ in terms of $|V(S_0)|,|V(S_1)|$ and $|V(S_2)|$.
\begin{proposition}{ \label{prop-another-bound-conditional-expectation} }
    With the same notations as in Lemma~\ref{lem-L^1-bound}, it holds that 
    \[
    \big|\mathbb E_{\Pb_\pi'}[ \phi_{S_1,S_2}]\big| \leq
    [1+o(1)]\cdot \rho^{|E(S_0)|}\big(nd^{20}\big)^{-\frac{|V(S_1)|+|V(S_2)|}{2}+|V(S_0)|} 2^{2+|E(S_1)|+|E(S_2)|-2|E(S_0)|} \,.
    \]
\end{proposition}
\noindent The proof of Proposition~\ref{prop-another-bound-conditional-expectation} is the most technical part of this paper and we postpone it into the appendix. 

With these results in hand, we can now complete the proof of Proposition~\ref{prop-controlling-expectation}. Firstly, it is straightforward to  check that under the same notations as in Lemma~\ref{lem-L^1-bound}, we have the following estimation (see Section~\ref{subsec-B3} for details):
    \begin{align}
        \min&\Big\{(4q)^{\frac{|E(S_1)|+|E(S_2)|}{2}-|E(S_0)|},\big(n d^{20}\big)^{-\frac{|V(S_1)|+|V(S_2)|}{2}+|V(S_0)|} 2^{2+|E(S_1)|+|E(S_2)|-2|E(S_0)|} \Big\} \nonumber \\
        &\qquad\qquad\quad\leq\ n^{-\frac{|V(S_1)|+|V(S_2)|}{2}+|V(S_0)|} d^{-7( |E(S_1)|+|E(S_2)|-2|E(S_0)| )} \,. \label{eq-quanlitative-estimation}
    \end{align}
For $S_0,S_1,S_2$ with $S_0\xhookrightarrow{} S_i,i=1,2$, define 
\[
M(S_0,S_1,S_2)=\rho^{|E(S_0)|}n^{-\frac{|V(S_1)|+|V(S_2)|}{2}+|V(S_0)|} d^{-7(|E(S_1)|+|E(S_2)|-2|E(S_0)|)} \,.
\]
Combining \eqref{eq-quanlitative-estimation} with Lemma~\ref{lem-L^1-bound} and Proposition~\ref{prop-another-bound-conditional-expectation} yields that conditioned on $\pi_*=\pi$, we have 
\begin{equation}{\label{eq-bound-fixed-pi-expectation}}
    \big|\mathbb E_{\Pb'_\pi}[\phi_{S_1,S_2}]\big|\le \big[1+o(1)\big]M(S_0,S_1,S_2)
\end{equation}
for any $S_1,S_2$ with at most $d$ edges and $S_0=S_1\cap \pi^{-1}(S_2)$. With these estimates in hand, we are ready to provide the proof of Proposition~\ref{prop-controlling-expectation}.
\begin{proof}[Proof of Proposition~\ref{prop-controlling-expectation}]
From \eqref{eq-relaxation-1}, by grouping the permutations $\pi\in \operatorname{S}_n$ according to the realization of $S_0=S_1\cap \pi^{-1}(S_2)$, we see $\big|\mathbb E_{\Pb'}[\phi_{S_1,S_2}]\big|$ is upper-bounded by
\begin{equation}\label{eq-relaxation-3}
    \frac{1}{n!}\sum_{\substack{S_0,S_0':S_0\cong S_0'\\S_0\subset S_1,S_0'\subset S_2}}\#\{\pi\in \operatorname{S}_n:S_0=S_1\cap \pi^{-1}(S_2),S_0'=\pi(S_0)\}\cdot M(S_0,S_1,S_2) \,.
\end{equation}
Note that for a permutation $\pi$ such that $\pi(S_0)=S_0'$, the vertices of $S_0$ must be mapped to that of $S_0'$, and the number of ways to map these vertices is given by $\operatorname{Aut}(S_0)$. As a result, we see for each $S_0$, it holds that
\begin{equation}\label{eq-relaxation-4}
\#\{\pi\in \operatorname{S}_n:S_0=S_1\cap\pi^{-1}(S_2),S_0'=\pi(S_0)\}\le\operatorname{Aut}(S_0)\cdot(n-|V(S_0)|)!\,.
\end{equation} 
Since $(n-|V(S_0)|)!/n!\le \big[1+o(1)\big]n^{-|V(S_0)|}$, we have that \eqref{eq-relaxation-3} is further upper-bounded by (up to a factor of $(1+o(1))$)
\begin{align*}
&\sum_{\substack{S_0,S_0':S_0\cong S_0'\\S_0\subset S_1,S_0'\subset S_2}}n^{-|V(S_0)|} \operatorname{Aut}(S_0)M(S_0,S_1,S_2)\\
=&\sum_{\substack{\mathbf H_0\in \mathcal H\\\mathbf H_0\xhookrightarrow{} S_i,i=1,2}} \frac{\operatorname{Aut}(\mathbf H_0)}{n^{|V(\mathbf H_0)|}} M(\mathbf H_0,S_1,S_2)\cdot \#\{(S_0,S_0'):S_0\subset S_1,S_0'\subset S_2,S_0\cong S_0'\cong \mathbf H_0\}\\
\le&\sum_{\substack{\mathbf H_0\in \mathcal H\\\mathbf H_0\xhookrightarrow{} S_i,i=1,2}} \frac{\operatorname{Aut}(\mathbf H_0)}{n^{|V(\mathbf H_0)|}} M(\mathbf H_0,S_1,S_2)\cdot d^{|E(S_1)|+E(S_2)-2|E(\mathbf H_0)|}\,,
\end{align*}
where the last inequality follows from Lemma~\ref{lemma-facts-graphs} (v) and the assumption that $|E(S_1)|$ and $|E(S_2)|$ are bounded by $d$. 
After rearranging the terms, this concludes \eqref{eq-control-expectation}.
\end{proof}

We are now ready to establish $\eqref{eq-evidence-detection}'$.
\begin{proof}[Proof of Theorem~\ref{thm-sparse-regime}]
    Recall from Lemma~\ref{lem-optimal-results-conditioning} that it suffices to show the optimal signal-to-noise ratio over $f\in \mathcal P_{n,d}'$ is $O(1)$. We estimate the optimal signal-to-noise ratio characterized by \eqref{eq-optimal-signal-to-noise-ratio-conditioning} via Proposition~\ref{prop-controlling-expectation} as follows:
    \begin{align*}
        & \sum_{\phi_{S_1,S_2}\in\mathcal{O}_{d}'} \big(\mathbb E_{\Pb'}[\phi_{S_1,S_2}]\big)^2 \\
        \le &\ [1+o(1)] \sum_{\substack{S_1,S_2\Subset\operatorname{K}_n\text{admissible}\\|E(S_1)|+|E(S_2)|\le d}} \Bigg( \sum_{\substack{\mathbf H_0\in\mathcal H\\\mathbf H_0\xhookrightarrow{} S_i,i=1,2}} \frac{ \rho^{|E(\mathbf H_0)|}\operatorname{Aut}(\mathbf H_0) }{ d^{6(|E(S_1)|+|E(S_2)|-2|E(\mathbf H_0)|)} n^{\frac{1}{2}(|V(S_1)|+|V(S_2)|)}  } \Bigg)^2 \\
        \leq &\ [1+o(1)] \sum_{\substack{S_1,S_2\Subset\operatorname{K}_n\text{admissible}\\|E(S_1)|+|E(S_2)|\le d}} \Bigg( \sum_{\substack{\mathbf H_0\in\mathcal H\\\mathbf H_0\xhookrightarrow{} S_i,i=1,2}} \frac{ \rho^{2|E(\mathbf H_0)|} \operatorname{Aut}(\mathbf H_0)^2 }{ d^{6(|E(S_1)|+|E(S_2)|-2|E(\mathbf H_0)|)} n^{|V(S_1)|+|V(S_2)|}  } \Bigg) \\
       & \times \Bigg( \sum_{\substack{\mathbf H_0\in\mathcal H\\\mathbf H_0\xhookrightarrow{} S_i,i=1,2}}d^{-6(|E(S_1)|+|E(S_2)|-2|E(\mathbf H_0)|)} \Bigg) \,,
    \end{align*}
    where the last inequality comes form Cauchy-Schwartz inequality.
    It is straightforward to check that when $|E(S_1)|+|E(S_2)|\le d$,
    \begin{align}
        \sum_{\substack{\mathbf H_0\in\mathcal H\\\mathbf H_0\xhookrightarrow{} S_i,i=1,2}}d^{-6(|E(S_1)|+|E(S_2)|-2|E(\mathbf H_0)|)} \le &\ \Big(\sum_{\substack{\mathbf H_0\in\mathcal H\\\mathbf H_0\xhookrightarrow{} S_1}}d^{-6(|E(S_1)|-|E(\mathbf H_0)|)}\Big)^2 \nonumber \\
        \le&\ \Big(\sum_{k=0}^d d^{-6k}\cdot d^k\Big)^2\leq 2 \,, \label{eq-xxx}
    \end{align}
    where the second inequality follows from Lemma~\ref{lemma-facts-graphs} (v). 
    Therefore, the optimal signal-to-noise ratio is upper-bounded by a constant factor times
    \begin{equation}\label{eq-relaxation-5} 
    \begin{aligned}
    & \sum_{\substack{S_1,S_2\Subset\operatorname{K}_n\text{admissible}\\|E(S_1)|+|E(S_2)|\le d}} \sum_{\substack{\mathbf H_0\in\mathcal H\\\mathbf H_0\xhookrightarrow{} S_i,i=1,2}}\frac{ \rho^{2|E(\mathbf H_0)|} \operatorname{Aut}(\mathbf H_0)^2 }{ d^{6(|E(S_1)|+|E(S_2)|-2|E(\mathbf H_0)|)} n^{|V(S_1)|+|V(S_2)|} } \\
    = & \sum_{
    \mathbf{H}_0\in\mathcal H } \rho^{2|E(\mathbf H_0)|} \operatorname{Aut}(\mathbf H_0)^2 \sum_{ \substack{ S_1,S_2\Subset\operatorname{K}_n\text{admissible}\\\mathbf H_0\xhookrightarrow{} S_i,i=1,2 \\ |E(S_1)|+|E(S_2)|\le d } } 
    \frac{1}{d^{6(|E(S_1)|+|E(S_2)|-2|E(\mathbf H_0)|)}n^{|V(S_1)|+|V(S_2)|}}  \,.
    \end{aligned}
    \end{equation}
    By denoting $S_1 \cong \mathbf{S}_1$ and $S_2 \cong \mathbf{S}_2$ with $\mathbf S_1,\mathbf S_2\in \mathcal H$, the right hand side of \eqref{eq-relaxation-5} reduces to
    \begin{equation}\label{eq-relaxation-6}
    \begin{aligned}
        & \sum_{ \mathbf{H}_0 \in \mathcal H} \rho^{2|E(\mathbf H_0)|} \operatorname{Aut}(\mathbf H_0)^2 \sum_{ \substack{ \mathbf S_1,\mathbf S_2\in \mathcal H\text{ admissible}\\\mathbf H_0\xhookrightarrow{} \mathbf{S}_i,i=1,2 \\ |E(\mathbf S_1)|+|E(\mathbf S_2)|\le d } } \frac{ \# \{ S_1,S_2 : S_1 \cong \mathbf{S}_1, S_2 \cong \mathbf{S}_2 \} }{ d^{6(|E(\mathbf S_1)|+|E(\mathbf S_2)|-2|E(\mathbf H_0)|)} n^{|V(\mathbf S_1)|+|V(\mathbf S_2)|} } \\
        \leq & \sum_{ \mathbf{H}_0 \in \mathcal H\operatorname{admissible} } \rho^{2|E(\mathbf H_0)|} \sum_{ \substack{ \mathbf S_1,\mathbf S_2\in \mathcal H,\mathbf H_0\xhookrightarrow{} \mathbf{S}_i,i=1,2 \\ |E(\mathbf S_1)|+|E(\mathbf S_2)|\le d } } \frac{ \operatorname{Aut}(\mathbf H_0)^2 }{ \operatorname{Aut}(\mathbf S_1) \operatorname{Aut}(\mathbf S_2) d^{ 6(|E(\mathbf S_1)|+|E(\mathbf S_2)|-2|E(\mathbf H_0)|) } } \,.
    \end{aligned}
    \end{equation}
    Using Lemma~\ref{lemma-facts-graphs} (iii), we have $\frac{ \operatorname{Aut}(\mathbf H_0)^2 }{ \operatorname{Aut}(\mathbf S_1) \operatorname{Aut}(\mathbf S_2) } \leq d^{ 2(|E(\mathbf S_1)|+|E(\mathbf S_2)|-2|E(\mathbf H_0)|) }$, and thus the right hand side of \eqref{eq-relaxation-6} is upper-bounded by
    \begin{equation}
    \begin{aligned}
        &\ \sum_{ \mathbf{H}_0\in \mathcal H \operatorname{admissible} } \rho^{2|E(\mathbf H_0)|} \sum_{ \substack{ \mathbf S_1,\mathbf S_2\subset \mathcal H, \mathbf H_0\xhookrightarrow{} \mathbf{S}_i,i=1,2 \\ |E(\mathbf S_1)|+|E(\mathbf S_2)|\le d } } d^{ -4(|E(\mathbf S_1)|+|E(\mathbf S_2)|-2|E(\mathbf H_0)|) } \\
        \leq&\ 4 \sum_{ \substack{\mathbf{H}_0\in \mathcal H \operatorname{admissible}\\ |E(\mathbf H_0)|\le d} } \rho^{2|E(\mathbf H_0)|} \,,
    \end{aligned}
    \end{equation}
    where we used the fact that (similar to \eqref{eq-xxx})
\begin{align*}
&\ \sum_{ \substack{ \mathbf S_1,\mathbf S_2\in \mathcal H, \mathbf H_0\xhookrightarrow{} \mathbf{S}_i,i=1,2 \\ |E(\mathbf S_1)|+|E(\mathbf S_2)|\le d } } d^{ -4(|E(\mathbf S_1)|+|E(\mathbf S_2)|-2|E(\mathbf H_0)|) }\\
\le&\  \Bigg(\sum_{\substack{\mathbf H_0\xhookrightarrow{} \mathbf S,|E(\mathbf S)|\le d}}d^{-4(|E(\mathbf S)|-|E(\mathbf H)|)}\Bigg)^2\le \Bigg(\sum_{k=0}^dd^{-4k}\cdot (2d)^{2k}\Bigg)^2\le 4\,.
\end{align*}
Furthermore, we have that 
    \begin{align}\label{eq-O(1)}
        \sum_{ \substack{\mathbf{H}_0 \in \mathcal H \operatorname{admissible}\\|E(\mathbf H_0)|\le d }} \rho^{2|E(\mathbf H_0)|} = \sum_{k=1}^{d} \rho^{2k} \# \big\{ \mathbf{H}_0 \in \mathcal H \operatorname{admissible} : |E(\mathbf H_0)|=k \big\} \,.
    \end{align}
    Now since $\rho^2<\alpha-\varepsilon$, we see from Lemma~\ref{lemma-enumerate-unlabeled-graph-bounded-density} that \eqref{eq-O(1)} is bounded by $2\alpha C/\varepsilon$ (here $C$ is the constant in the statement of Lemma~\ref{lemma-enumerate-unlabeled-graph-bounded-density}) for any sufficiently large $n$. This shows that the optimal signal-to-noise ratio is uniformly bounded for all $n$ and thus proves $\eqref{eq-evidence-detection}'$.   
\end{proof}

\appendix
\section{Preliminary facts about graphs}

\begin{lemma}{\label{lemma-facts-graphs}}
Let $S,T\Subset \operatorname{K}_n$ satisfy $S\cong\mathbf{S}$ and $T\cong\mathbf{T}$ for some $\mathbf{S},\mathbf{T}\in \mathcal H$. Recall $S\cup T,S\cap T\Subset \operatorname{K}_n$ defined as edge-induced graphs of $\operatorname{K}_n$. We have the following hold:
\begin{enumerate}
    \item[(i)] $|V(S\cup T)|+|V(S\cap T)|\le|V(S)|+|V(T)|$, and $|E(S\cup T)|+|E(S\cap T)|=|E(S)|+|E(T)|$.
    \item[(ii)] $\Phi(S\cup T)\Phi(S\cap T)\le\Phi(S)\Phi(T)$.
    \item[(iii)] If $\mathbf{S}\subset\mathbf{T}$, then $|\!\operatorname{Aut}(\mathbf{S})|\leq |\!\operatorname{Aut}(\mathbf{T})|\cdot |V(\mathbf{T})|^{ 2(|E(\mathbf{T})|-|E(\mathbf{S})|) }$.
    \item[(iv)] $\#\big\{ T'\Subset \operatorname{K}_n:S\subset T',|V(T')|-|V(S)|=k,|E(T')|-|E(S)|=l \big\} \leq n^{k}(|V(S)|+k)^{2l}$.
    \item[(v)] $\#\big\{ T' \subset S : |E(S)|-|E(T')|=k \big\} \le |E(S)|^k$.
\end{enumerate}
\end{lemma}

{\begin{proof}
    By definition, we have $V(S\cup T)=V(S)\cup V(T)$, $E(S\cup T)=E(S)\cup E(T)$ and $E(S\cap T)=E(S) \cap E(T)$. In addition, we have $V(S\cap T) \subset V(S) \cap V(T)$ because for any $i \in V(S \cap T)$, there exists some $j$ such that $(i,j) \in E(S\cap T)$ and thus $i \in V(S) \cap V(T)$. Therefore, (i) follows from the inclusion-exclusion formula.
    Building on this, (ii) follows directly from \eqref{eq-def-Phi}.
    
    For (iii), let $\mathcal I=\big\{H:V(H)= \{1,\ldots,|V(\mathbf{T})|\}, H\cong\mathbf{T}\big\}$, let $\mathcal{J}=\big\{H_0:V(H_0)\subset\{1,\ldots,|V(\mathbf{T})|\}, H_0\cong\mathbf{S}\big\}$, and let $\mathcal{E}=\big\{E\subset E(1,\ldots,|V(\mathbf{T})|):|E|=|E(\mathbf{T})|-|E(\mathbf{S})|\big\}$. For any $H\in\mathcal{I}$, since $H\cong\mathbf{T}$ we can take an arbitrary but prefixed subgraph $H_0\subset H$ with $H_0\cong\mathbf{S}$, and then we can define a mapping $\Gamma: \mathcal I \to \mathcal{J}\times\mathcal{E}$ by
    \begin{align*}
        H \overset{ \Gamma }{ \longrightarrow } (H_0, E(H) \setminus E(H_0)) \,.
    \end{align*}
    It is easy to check $\Gamma$ is an injection. Combining with the fact that $$|\mathcal{E}|\leq \binom{|V(\mathbf{T})|}{2}^{|V(\mathbf{T})|-|V(\mathbf{S})|}\leq |V(\mathbf{T})|^{2(|E(\mathbf{T})|-|E(\mathbf{S})|)}\,,$$ we conclude that
    \begin{align*}
        |\mathcal{I}|=\frac{|V(\mathbf{T})|!}{|\!\operatorname{Aut}(\mathbf{T})|} \leq|\mathcal J|\cdot|\mathcal{E}|\leq \frac{|V(\mathbf{T})|!}{(|V(\mathbf{T})|-|V(\mathbf{S})|)!|\!\operatorname{Aut}(\mathbf{S})|}|V(\mathbf{T})|^{2(|E(\mathbf{T})|-|E(\mathbf{S})|)} \,,
    \end{align*}
    which implies $\frac{|\!\operatorname{Aut}(\mathbf{S})|}{|\!\operatorname{Aut}(\mathbf{T})|} \leq |V(\mathbf{T})|^{2(|E(\mathbf{T})|-|E(\mathbf{S})|)}$, as desired.
    
    For (iv), we can produce all subgraphs $T' \Subset \operatorname{K}_n$ such that $S\subset T',|V(T')|-|V(S)|=k,|E(T')|-|E(S)|=l$ via the following procedure: (1) we choose the vertex set $V(T')\setminus V(S) \subset [n]$, and the number of such choices is at most $n^k$; (2) we choose the edge set $E(T')\setminus E(S)$ within $V(T')$, and the number of such choices is at most $|V(T')|^{2(|E(T')|-|E(S)|)}=(|V(S)|+k)^{2l}$. This completes the proof by applying the multiplication principle.

    For (v), it follows from the fact that we can produce all $T' \subset S$ with $|E(S)|-|E(T)|=k$ by deleting $k$ edges in $E(S)$.
\end{proof}

The following two lemmas pertain to admissible graphs. It is crucial to keep in mind that our definition of admissible graphs depends on both the parameter $n$ and the specific choices of $d = d_n$ and $q = q_n$. 
\begin{lemma}{\label{lem-density-admissible-graph}}
    For a graph $\mathbf H \in \mathcal H$, we define its maximal edge density to be $\max_{\mathbf H' \subset \mathbf H} \frac{|E(\mathbf H')|}{|V(\mathbf H')|}$. Then for sufficiently large $n$, every admissible graph has maximal edge density bounded by $1+\frac{8}{d}+\frac{\varepsilon}{10\log d}$.
\end{lemma}
\begin{proof}
    For any admissible graph $\mathbf{H}$ and any $\mathbf H' \subset \mathbf{H}$, using Definition~\ref{def-addmisible} we have
    \begin{align*}
    \frac{|E(\mathbf H')|}{|V(\mathbf H')|} &\leq \frac{(1+4/d)\log n+20\log d}{\log q^{-1}-6\log d}+\frac{\log\log n}{|V(\mathbf H')|(\log q^{-1}-6\log d)} \\
    &\leq 1+\frac{8}{d}+\frac{40\log d}{\log n}+\frac{\log nq}{\log n}+\frac{2\log\log n}{\log n} \le 1+\frac{8}{d}+\frac{\varepsilon}{10\log d}\,,
    \end{align*}
    provided that $n$ is large enough (the last inequality used ${\log d = o\big( \frac{\log n}{\log nq}\wedge  \sqrt{\log n}\big)}$). 
\end{proof}

\begin{lemma} {\label{lemma-enumerate-unlabeled-graph-bounded-density}}
For any fixed constant $\varepsilon\in(0,0.1)$, there exists a constant $C>0$ depending only on $\varepsilon$ such that for sufficiently large $n$, the number of isomorphic classes of admissible graphs with $N \le d$ edges are bounded by $C\cdot({\alpha}-\varepsilon/2)^{-N}$, where $\alpha$ is the Otter's constant.
\end{lemma}
\begin{proof}
Using Lemma~\ref{lem-density-admissible-graph}, for $n$ large enough, all admissible graphs with $N\le d$ edges have maximal edge density bounded by $1+\frac{8}{d}+\frac{\varepsilon}{10\log d}\le 1+\frac 8N+\frac{\varepsilon}{10\log N}$. Thus, denoting $\delta=\frac{8}{N}+\frac{\varepsilon}{10\log N}$ and letting $\mathcal A_{\delta,N}$ be the collection of (isomorphic classes of) graphs $\mathbf H\in \mathcal H$ with $N$ edges and maximal edge density at most $\delta$, it suffice to show that $|\mathcal{A}_{\delta,N}| \le C({\alpha}-\varepsilon/2)^{-N}$ for some constant $C$.

For a positive real number $L$ to be chosen later, let $\mathcal{H}_{L}=\{\mathbf J\in \mathcal H:|E(\mathbf J)|\le L\}$. For any $\mathbf H \in \mathcal{A}_{\delta,N}$, we perform the following procedure:
\begin{itemize}
    \item Decompose $\mathbf H$ into connected components $\mathbf H= \mathbf U_1 \sqcup \mathbf U_2 \ldots \sqcup \mathbf U_{k}$, and define the vector $\{ x(\mathbf J)\}_{ \mathbf J \in \mathcal{H}_{L} }$ given by $x(\mathbf J)=\#\{ 1 \leq i \leq k: \mathbf U_i \cong \mathbf J \}$.
    \item Without loss of generality, we may assume that $\mathbf U_1,\dots,\mathbf U_m$ are the components that have at least $L$ edges with $X_i=|E(\mathbf U_i)|$ listed in the decreasing order. From the definition of $\mathcal{A}_{\delta,N}$ we also know that $(1+\delta)^{-1} X_i \leq |V(\mathbf U_i)| \leq X_i+1$.
    \item For each $1 \leq i \leq m$, we choose a spanning tree $\mathbf T_i$ of $\mathbf U_i$, and denote $\mathbf E_i = \mathbf U_i \setminus \mathbf T_i$. Clearly $|V(\mathbf T_i)|=|V(\mathbf U_i)|\le X_i+1$ and $|E(\mathbf E_i)| \leq X_i-|V(\mathbf{U_i})|+1 \leq \delta X_i+1 $. 
\end{itemize}
From the above procedure, we get a triple $\big( \{ x(\mathbf J) : \mathbf J \in \mathcal{H}_{L} \}, (\mathbf{E}_{1}, \ldots, \mathbf{E}_{m}), (\mathbf T_{1}, \ldots, \mathbf T_{m}) \big)$ which uniquely encodes $\mathbf H\in \mathcal A_{\epsilon,N}$ such that the following hold:
\begin{itemize}
    \item $\sum_{ \mathbf J \in \mathcal{H}_{L} } x(\mathbf J) \leq N$;
    \item $|V(\mathbf E_i)| \leq X_i+1$ and $|E(\mathbf E_i)| \leq \delta X_i+1$ for $1\leq i \leq m$;
    \item $|V(\mathbf{T}_i)| \leq X_i+1$ for $1 \leq i \leq m$.
\end{itemize}
We can thus upper-bound $|\mathcal{A}_{\delta,N}|$ by controlling the enumeration of such triples. Since $\sum_{\mathbf J \in \mathcal{H}_{L}} x(\mathbf J) \leq N$ and we have a trivial bound 
$$| \mathcal{H}_{L} | \leq \binom{2L}{2}^{L} \leq 2^{2L^2}\,,$$
we see that the number of choices of $\{ x(\mathbf J) : \mathbf J \in \mathcal{H}_{L} \}$ is upper-bounded by
\begin{align*}
    N^{ | \mathcal{H}_{L} | } = \exp \big( |\mathcal{H}_{L}| \log N \big) \leq \exp \big(2^{2L^2} \log N \big) \,.
\end{align*}
In order to bound the enumeration of $(\mathbf E_1,\dots,\mathbf E_m)$ and $(\mathbf T_1,\dots,\mathbf T_m)$, we first fix $X_{i}=|V(\mathbf U_i)|$ for $1\le i\le m$. Clearly the number of choices of $\{ \mathbf{E}_i: 1 \leq i \leq m\}$ is upper-bounded by $\prod_{i=1}^{m} (X_i^2)^{\delta X_i+1}$. Moreover, from \cite{Otter48} there exists a universal constant $C_0$ such that for any $V\in \mathbb N$, the number of unlabeled trees with $V$ vertices is upper-bounded by $e^{C_0} \alpha^{-V}$. Hence (noticing the trivial bound $m \leq \frac{N}{L}$) the number of choices of $(\mathbf T_{1}, \ldots, \mathbf{T}_m)$ is upper-bounded by
\begin{align*}
    \prod_{i=1}^{m} \big( e^{C_0} \alpha^{-X_i-1} \big) \leq e^{C_0 m} \alpha^{-N-m} \leq e^{\frac{C_0 N}{L}} \alpha^{-(N+\frac{N}{L})} \,.
\end{align*}
Therefore, we get that
\begin{align*}
    | \mathcal{A}_{\delta,N} | \leq & \alpha^{-N-\frac{N}{L} } \exp \left( 2^{2L^2}\log N + \frac{C_0 N}{L} \right) \cdot \sum_{m \geq 0} \sum_{ \substack{ X_1 \geq \ldots \geq X_m \geq L \\ X_1 + \ldots + X_m \leq N }  } \prod_{i=1}^{m} X_i^{ 2(\delta X_i+1)} \,.
\end{align*}
To bound the last summation, we define $M_k = \# \{ X_i : X_i=k \}$, and thus we have that $\sum_{k=L}^{N} k M_k \leq N$. Using the fact that $\frac{(\delta k+1)\log k}{k} \leq \delta \log N + \frac{\log L}{L} $ when $L \leq k \leq N$, we have
\begin{align*}
    &\prod_{i=1}^{m} X_i^{2(\delta X_i+1)} = \prod_{k=L}^{N} k^{2(\delta k+1) M_k} \leq \exp \Big\{ 2\sum_{k=L}^{N} M_k \big(\delta k+1 \big)\log k  \Big\} \\
    \leq &\exp \bigg( 2\Big( \delta \log N + \frac{\log L}{L}\Big) \sum_{k} k M_k \bigg) \leq \exp \left(2\delta N \log N + \frac{ 2N \log L }{L} \right) \,.
\end{align*}
Thus, by choosing $L=\sqrt[3]{\log N}$ we get that 
\begin{align*}
    |\mathcal A_{\delta,N}|\le &\ \alpha^{-N-\frac{N}{\sqrt[3]{\log N}}} \exp \left( \frac{C_0 N}{\sqrt[3]{\log N}} + 4^{\sqrt[3]{(\log N)^2}} \log N + 2 \delta N \log N + \frac{N \log \log N}{\sqrt[3]{\log N}} \right)\times \\
    &\ \# \Big\{ (m,X_1,\ldots,X_m) : X_1 \geq X_2 \geq \ldots \geq X_m \geq L, X_1 + \ldots + X_m \le N \Big\} \\
    \le&\ \alpha^{-N}\cdot\exp\left(\frac{C_1 N\log\log N}{\sqrt[3]{\log N}}+2\delta N\log N\right)\cdot\sum_{N'=1}^N\#\{\text{partitions of }N'\}\,,
\end{align*}
where $C_1$ is also an universal constant.
Finally, from a well-known result due to Siegel (see e.g. \cite[pp.316-318]{Aposto76}), the number of partitions of $N'$ is no more than $\exp\big(10\sqrt{N'}\big)$. In conclusion, we have shown that
\begin{align*}
|\mathcal{A}_{\delta,N}| &\le \alpha^{-N} \cdot \exp \left( \frac{C_1 N \log \log N}{\log N} + 2\Big(\frac{8}{N}+\frac{\varepsilon}{10\log N} \Big)N\log N \right) \cdot N \exp \left( 10 \sqrt{N} \right) \\
&= \alpha^{-N} \exp \left( \frac{\varepsilon N}{5} + \frac{C_1 N \log \log N}{\log N} + 17\log N + 10\sqrt{N} \right) \,,
\end{align*}
which is upper-bounded by $C\cdot({\alpha}-\varepsilon/2)^{-N}$ for any $N$ for some large enough (universal) constant $C$ since $e^{\varepsilon/5}<(1-\frac{\varepsilon}{2{\alpha}})^{-2}$.
\end{proof}

\section{Complimentary proofs}
\subsection{Proof of Proposition~\ref{prop-another-bound-conditional-expectation}.}
This section is dedicated to proving Proposition~\ref{prop-another-bound-conditional-expectation}, and we begin by introducing some notations. Recall Definition~\ref{def-correlated-random-graph}, where $I_e=\mathbf{1}_{e\in E(G)}$ for $e\in \operatorname{U}$. We denote $\widehat\Pb=\mathbf B(1,p)^{\otimes \!\operatorname{U}}$ as the marginal distribution of $\{I_e:e\in \operatorname{U}\}$ under $\Pb_*$. Furthermore, let $\mathcal F_I=\sigma(\{I_e:e \in \operatorname{U}\})$ be the corresponding $\sigma$-field. It is important to note that $\mathcal F_I$ is independent of $\pi_*$. For clarity, we will use $\widehat{\mathbb E}$ and $\widehat{\mathbb E}'$ to denote the expectation with respect to $\widehat\Pb$ and $\widehat\Pb[\cdot\mid \Gc]$, respectively (recall that $\mathcal{G}$ is measurable with respect to $\mathcal F_I$).

Throughout this section, we will consider a fixed permutation $\pi$ and always condition on $\pi_*=\pi$. Let $\Pb_\pi$ and $\Pb'_\pi$ denote $\Pb_*[\cdot\mid \pi_*=\pi]$ and $\Pb_*'[\cdot\mid \pi_*=\pi]$, respectively. Additionally, we use $\mathbb E_\pi$ and $\mathbb E'_\pi$ to represent the expectation under $\Pb_\pi$ and $\Pb'_\pi$, respectively. Applying the iterated expectation formula by first conditioning on $\mathcal F_I$, we have
\begin{align*}
    & \mathbb{E}_\pi'\big[ \phi_{S_1,S_2}\big] = \frac{1}{\widehat{\mathbb{P}}[\mathcal{G}]}\cdot\widehat{\mathbb{E}} \Big[ \mathbf{1}_{\mathcal{G}} \cdot\mathbb{E}_\pi  \big[ \phi_{S_1,S_2} \mid \mathcal F_I \big] \Big] = \widehat{\mathbb{E}}' \Big[ \mathbb{E}_\pi \big[ \phi_{S_1,S_2} \mid \mathcal F_I \big] \Big] \,.
\end{align*}
Clearly by definition we have that for every $e\in \operatorname{U}$,
$$\mathbb{E}_\pi[A_e \mid \mathcal F_I]=\mathbb{E}_\pi[B_{\pi(e)}\mid \mathcal F_I]=sI_e,\quad\mathbb{E}_\pi[A_e B_{\pi(e)}\mid \mathcal F_I]=s^2I_e\,,$$
and thus the conditional expectation $\mathbb{E}_{\pi} \big[ \phi_{S_1,S_2}(A,B) \mid \mathcal F_I \big]$ is given by $h(S_1,S_2)$, where (recall that $S_0=S_1 \cap\pi^{-1}(S_2)$)
\begin{align}
    h(S_1,S_2)
    = \prod_{ e \in E(S_1) \triangle E(\pi^{-1}(S_2)) } \frac{ s I_e - q }{ \sqrt{q(1-q)} } \prod_{e \in E(S_0)} \frac{ s^2 I_e - 2qs I_e + q^2 }{ q(1-q) } \,.
    \label{eq-phi-S_1-S_2-given-I-pi}
\end{align}
The main technical input is the following result.
\begin{lemma}{ \label{lemma-conditional-expectation-edge-products} }
For any graph $S\Subset \operatorname{K}_n$ such that $|V(S)|\le 2d$,  denote $\mathcal F_S\subset \mathcal F_I$ the $\sigma$-field generated by $\{I_e:e\in E(S)\}$. In addition,  let $\mathsf E$ be a collection of edges with $|\mathsf E|\leq d$ and $\mathsf E \cap E(S) = \emptyset$, and let $V(\mathsf E)$ be the set of end points of edges in $\mathsf E$. Then for any $\mathcal F_S$-measurable bounded function $\varphi$, it holds that
\begin{equation}{ \label{eq-conditional-expectation-edge-products} }
    \Big| \widehat{\mathbb{E}}' \big[ \varphi \prod_{\mathsf e\in \mathsf E} ( s I_{\mathsf e} - q ) \big] \Big| \leq [1+o(1)]\cdot 2^{|\mathsf E|+2} \widehat{\mathbb{E}} \big[ | \varphi| \big] \cdot \big(n d^{20}\big)^{-|V(\mathsf E) \setminus V(S)|} \,.
\end{equation}
\end{lemma}
\noindent Assuming Lemma~\ref{lemma-conditional-expectation-edge-products}, we next prove Proposition~\ref{prop-another-bound-conditional-expectation}.
\begin{proof}[Proof of Proposition~\ref{prop-another-bound-conditional-expectation}]
For any $S_1,S_2\Subset \operatorname{K}_n$, we define the function
\begin{align}
    \varphi_1(S_1,S_2)= \prod_{e \in E(S_1) \setminus E(S_0)} \frac{s I_e-q}{\sqrt{q(1-q)}} \prod_{e \in E(S_0)} \frac{s^2 I_e-2qs I_e+q^2}{q(1-q)} \,,\label{eq-def-varphi-1}
\end{align}
which is bounded and $\mathcal F_{S_1}$-measurable. It is clear that
\begin{align*}
    h(S_1,S_2)=\varphi_1(S_1,S_2) \cdot \prod_{e\in E(\pi^{-1}(S_2))\setminus E(S_0)} \frac{s I_e-q}{\sqrt{q(1-q)}} \,.
\end{align*}
Therefore, we have
\begin{align*}
    &\big| \mathbb{E}'_\pi[ h(S_1,S_2) ] \big|=\Big| \widehat{\mathbb{E}}'\big[\varphi_1(S_1,S_2) \prod_{e\in E(\pi^{-1}(S_2))\setminus E(S_0)}(sI_e-q) \big]\Big|\cdot \left(\frac{1}{\sqrt{q(1-q)}}\right)^{|E(S_2)|-|E(S_0)|} \\\leq 
    &\ [1+o(1)] \cdot \widehat{\mathbb{E}}\big[|\varphi_1(S_1,S_2)|\big] 2^{2+|E(S_2)|-|E(S_0)|} \big(n d^{20}\big)^{-|V(\pi^{-1}(S_2)) \setminus V(S_0)|} q^{\frac{1}{2}(|E(S_0)|-|E(S_2)|)} \,,
\end{align*}
where the inequality follows from applying Lemma~\ref{lemma-conditional-expectation-edge-products} by taking $\varphi=\varphi_1(S_1,S_2)$ and $\mathsf E=E(\pi^{-1}(S_2)) \setminus E(S_0)$. 
Recalling \eqref{eq-def-varphi-1}, we get from a straightforward calculation that
\begin{align*}
    &\ \widehat{\mathbb{E}} \big[|\varphi_1(S_1,S_2)|\big] = \prod_{e \in E(S_0)} \mathbb{E} \Big[ \frac{s^2I_e-2qs I_e+q^2}{q(1-q)} \Big] \prod_{e \in E(S_1) \setminus E(S_0)} \mathbb{E} \Big[\frac{|s I_e-q|}{\sqrt{q(1-q)}}\Big] \\
    \leq &\ \rho^{|E(S_0)|} \Big(\frac{2\sqrt{q}}{\sqrt{1-q}} \Big)^{|E(S_1)|-|E(S_0)|} = [1+o(1)] \rho^{|E(S_0)|} 2^{|E(S_1)|-|E(S_0)|}q^{\frac{1}{2}(|E(S_1)|-|E(S_0)|)} \,.
\end{align*}
Combining the preceding two displays, we get $\big| \mathbb{E}'_\pi[ h(S_1,S_2)] \big|$ is upper-bounded by $1+o(1)$ times
\begin{align}
    \rho^{|E(S_0)|} 2^{2+|E(S_1)|+|E(S_2)|-2|E(S_0)|} q^{\frac{1}{2}(|E(S_1)|-|E(S_2)|)} \big(n d^{20}\big)^{-|V(S_2)|+|V(S_0)|}  \,.  \label{eq-possible-bound-h-I}
\end{align}
Similarly, we can define 
\begin{align*}
    \varphi_2(S_1,S_2)= \prod_{e \in E(\pi^{-1}(S_2)) \setminus E(S_0)} \frac{s I_e-q}{\sqrt{q(1-q)}} \prod_{e \in E(S_0)} \frac{s^2I_e-2qs I_e+q^2}{q(1-q)}
\end{align*}
and get that $\big| \mathbb{E}_\pi'[ h(S_1,S_2)] \big|$ is also bounded by $1+o(1)$ times
\begin{align}
    \rho^{|E(S_0)|} 2^{2+|E(S_1)|+|E(S_2)|-2|E(S_0)|} q^{\frac{1}{2}(|E(S_2)|-|E(S_1)|)} \big(n d^{20}\big)^{-|V(S_1)|+|V(S_0)|}   \,.  \label{eq-possible-bound-h-II}
\end{align}
Combining the two upper bounds \eqref{eq-possible-bound-h-I} and \eqref{eq-possible-bound-h-II} yields the desired result.
\end{proof}

The rest of this section is devoted to the proof of Lemma~\ref{lemma-conditional-expectation-edge-products} and henceforth we will only deal with $\widehat\Pb$ and $\widehat\Pb'$. For $W\subset [n]$, let $\operatorname{K}_{W}$ be the complete graph on the vertex set $W$. We will fix the choices of $S$ and $\mathsf E$, and without loss of generality, we may assume that $|V(S)| \geq 10$ and $S$ contains all the edges in $\operatorname{K}_{V(S)} \setminus \mathsf E$, since otherwise we can always replace $S$ by a larger graph $S' \supset S$ such that $|V(S')|\le 2d$ and $V(S) \cap V(\mathsf E)=V(S') \cap V(\mathsf E)$.

Fix any $\mathcal F_S$-measurable bounded function $\varphi$. We expand the left hand side of \eqref{eq-conditional-expectation-edge-products} as 
\begin{align}\label{eq-expansion}
    \sum_{\chi_S} \sum_{\chi_{\mathsf E} } \varphi(\chi_S) \prod_{\mathsf e\in \mathsf E} (s\chi_{\mathsf e}-q) \cdot \widehat{\mathbb{P}}'[ I_{\mathsf E}=\chi_{\mathsf E}, I_S=\chi_S ] \,,
\end{align}
where the summation of $\chi_S$ (respectively, $\chi_{\mathsf E}$) is taken over all possible realizations of $I_{S} = \{ I_{e} : e \in E(S) \} \in \{ 0,1\}^{E(S)}$ (respectively, $I_{\mathsf E}=\{ I_{\mathsf e}: \mathsf e\in \mathsf E\}\in \{0,1\}^{\mathsf E}$). Using Bayes formula and the independence between $I_S$ and $I_{\mathsf E}$ under $\widehat\Pb$, \eqref{eq-expansion} equals 
\begin{align*}
    & \sum_{\chi_S} \sum_{ \chi_{\mathsf E} } \widehat\Pb[I_S=\chi_S] \varphi(\chi_S) \prod_{e\in \mathsf E} \Big( (s\chi_{e}-q) \widehat{\mathbb{P}}[ I_{e}= \chi_{e} ] \Big) \frac{ \widehat{\mathbb{P}}[\mathcal{G} \mid I_{\mathsf E}=\chi_{\mathsf E} , I_S=\chi_S ] }{ \widehat{\mathbb{P}}[\mathcal{G}] } \\
    =&\ \frac{(q(1-p))^{|\mathsf E|}}{1-o(1)} \sum_{\chi_S} \widehat{\mathbb{P}}[I_S = \chi_S] \varphi(\chi_S) \sum_{ \chi_{\mathsf E} } (-1)^{\|\chi_{\mathsf E}\|_0} \widehat{\mathbb{P}}[ \mathcal{G} \mid I_{\mathsf E}=\chi_{\mathsf E}, I_{S} = \chi_S ] \,,
\end{align*}
where $\|\chi_{\mathsf E}\|_0$ denotes the number of zero coordinates in the vector $\chi_{\mathsf E}$.
We say $\chi_S$ is self-admissible, if the graph with edge set $\{ e \in E(S): \chi_e =1 \}$ is admissible. Denoting $\mathsf A_S\subset \{0,1\}^{E(S)}$ as the set of all self-admissible tuples, we then see that one only needs to sum over $\chi_S\in \mathsf A_S$ (since otherwise the conditional probability of $\mathcal G$ is $0$). Define $\Omega$ to be the set of all possible realizations of $\{ I_{e} : e \not \in \mathsf E\cup E(S) \} \in \{ 0,1 \}^{ \operatorname{U} \setminus (\mathsf E\cup E(S))}$, we have
\begin{equation}{ \label{eq-conditional-probability-mathcal-G} }
    \widehat{\mathbb{P}}[\mathcal{G} \mid I_{\mathsf E}= \chi_{\mathsf E}, I_{S} = \chi_S ] = \sum_{\omega \in \Omega} \widehat{\mathbb{P}}[\omega] \mathbf{1}_{ \omega \oplus \chi_{\mathsf E}   \oplus \chi_S   \in \mathcal{G} } \,,
\end{equation}
where $\omega \oplus \chi_{\mathsf E}  \oplus \chi_S  $ represents the realization of $\{ I_{e} : e \in \operatorname{U} \}$ such that 
\begin{align*}
    I_{e} = 
    \begin{cases}
        \chi_{e}, & e \in \mathsf E\cup E(S) ; \\
        \omega_e, & \mbox{otherwise} \,.
    \end{cases}
\end{align*}
Applying \eqref{eq-conditional-probability-mathcal-G}, we get that
\begin{align}
    \sum_{ \chi_{\mathsf E} } (-1)^{\|\chi_{\mathsf E}\|_0} \widehat{\mathbb{P}}[\Gc\mid I_{\mathsf E}=\chi_{\mathsf E}, I_{S} = \chi_S ] 
    =\sum_{\omega} \widehat{\mathbb{P}}[\omega]\sum_{ \chi_{\mathsf E} } (-1)^{\|\chi_{\mathsf E}\|_0} \mathbf{1}_{ \omega \oplus \chi_{\mathsf E}  \oplus \chi_S   \in \mathcal{G} } \,. \label{eq-reorganized-summation}
\end{align}
We say $\omega$ is an {\em improper realization} with respect to ${\mathsf E}$ and $\chi_S$ (denoted as $\omega\sim({\mathsf E},\chi_S)$), if
\begin{equation}{ \label{eq-def-proper-realization} }
    \sum_{ \chi_{\mathsf E} } (-1)^{\|\chi_{\mathsf E}\|_0} \mathbf{1}_{ \omega \oplus \chi_{\mathsf E}  \oplus \chi_S   \in \mathcal{G}   } \neq 0 \,. 
\end{equation}
Clearly \eqref{eq-reorganized-summation} is upper-bounded by 
$
    2^{|\mathsf E|}\cdot \widehat{\mathbb{P}}[ \omega\sim( {\mathsf E}, \chi_S) ]
$,
and thus we have 
\begin{align}
    \Big|\widehat{ \mathbb{E}}'[ \varphi(\chi_S) \prod_{e\in \mathsf E}(s I_{e}-q) ] \Big|\le &\ [1+o(1)](2q(1-p))^{|\mathsf E|} \sum_{\chi_S \in\mathsf A_S} |\varphi(\chi_S)| \widehat{\mathbb{P}}[I_S=\chi_S] \widehat{\mathbb{P}}[\omega \sim({\mathsf E}, \chi_S)]  \nonumber \\
    \leq&\  [1+o(1)] (2q)^{|\mathsf E|} \mathbb{E} \big[|\varphi(\chi_S)|\big] \max_{\chi_S\in\mathsf A_S} \widehat{\mathbb{P}}[ \omega \sim(\mathsf E,\chi_S)] \,. 
\end{align}
It remains to show the following (somewhat mysterious) lemma:
\begin{lemma}{ \label{lemma-prob-proper-realization} }
    Given $|\mathsf E| \leq d$ and $|V(S)| \leq 2d$, for all $\chi_S\in \mathsf A_S$ we have 
    \begin{equation*}
        \widehat{\mathbb{P}}[ \omega \sim({\mathsf E},\chi_S)] \leq 4\big(n d^{20}\big)^{-|V(\mathsf E) \setminus V(S)|}  q^{-|\mathsf E|}\,.
    \end{equation*} 
\end{lemma}
To prove Lemma~\ref{lemma-prob-proper-realization}, we need to further specify another desirable property of $\omega$. For simplicity, in what follows we will use $\mathsf{1}_{\mathsf E}$ and $\mathsf{0}_{\mathsf E}$ to denote the all-one and all-zero vector in $\{0,1\}^{\mathsf E}$, respectively. 
\begin{DEF}\label{def-Ec}
    Let $\mathcal{E}=\mathcal E({\mathsf E},\chi_S)$ be the event that in the graph with edge indicators $\omega\oplus\mathsf{1}_{\mathsf E}\oplus\chi_S$, there is no bad subgraph $H$ with $d^{2}<|V(H)|\leq 2d^2$. 
    In addition, for $\omega\sim({\mathsf E},\chi_S)$, if $\omega$ also belongs to $\mathcal E$, then we write $\omega\stackrel{*}{\sim}({\mathsf E},\chi_S)$. 
\end{DEF}
\begin{remark}
The event $\mathcal E$ is introduced for the sake of convenience when taking union of two graphs in subsequent proof. Basically, the argument will require a relationship of $1:2$ between the lower and upper bounds of $|V(H)|$ in $\mathcal E$. Here the choice of $d^2$ as the lower bound is for the purpose of a nice upper bound on $\widehat{\Pb}[\mathcal E^c]$ as given in \eqref{eq-bound-on-P[Ec]}. The introduction of $\mathcal E$ also partly explains why our truncation event $\Gc$ involves subgraphs up to size $d^2$.
 \end{remark}
 It is clear that
\[
    \widehat{\mathbb P}[ \omega \sim({\mathsf E},\chi_S) ] \leq \widehat{\mathbb{P}}[\mathcal{E}^c] + \widehat{\mathbb P}[\omega \stackrel{*}{\sim}({\mathsf E},\chi_S) ] \,.
\]
Similar as in Lemma~\ref{eq-Gc-is-typical}, one can show by taking a union bound that for any $\chi_S\in \mathsf A_S$ and $\mathsf E$ with $|{\mathsf E}| \le d$, 
\begin{equation}\label{eq-bound-on-P[Ec]}
\widehat{\mathbb{P}}[\mathcal E^c] \leq 2^{|V(S)|} n^{-4d^2/d} d^{-20d^2}  q^{-|\mathsf E|} \leq \big(nd^{20}\big)^{-4d}  q^{-|\mathsf E|}\,.
\end{equation}
Therefore, 
it remains to show
\begin{equation} \label{eq-upp-bound-truncate-prob-proper-realiaztion} 
    \widehat{\mathbb P}[ \omega \stackrel{*}{\sim}({\mathsf E},\chi_S)] \leq 2\big(n d^{20}\big)^{-|V(\mathsf E) \setminus V(S)|}  q^{-|\mathsf E|} \,,\ \forall\chi_S\in \mathsf A_S\,.
\end{equation}

Henceforth we also fix the realization of $\chi_S\in \mathsf A_S$. For convenience of presentation, we next introduce some more notations. Given any vector $\tau=\{ \tau_{e}\}_{e\in \operatorname{U}}\in {[0,\infty)}^{\operatorname{U}}$ and any $W\subset [n]$, define 
\begin{equation}{ \label{equ-def-Phi-tau} }
\Phi_{\tau}(W) = \big(n^{1+4/d}d^{20}\big)^{|W|} \big(qd^6\big)^{\sum_{e \in \operatorname{K}_W} \tau_e}\,.
\end{equation}
Here we remark that $\Phi_\tau$ can be regarded as an analogue of $\Phi$ defined as in Definition~\ref{def-addmisible} for induced subgraphs of a weighted graph on $[n]$. Our definition involves vectors in ${[0,\infty)}^{\operatorname{U}}$ (understood as weights on the edges) for technical reasons which will become clear later, and conceptually one may have in mind that $\tau \in \{0, 1\}^{\operatorname{U}}$ is the vector of edge indicators of some graph on $[n]$. It would also be useful to note that for any $\tau \in {[0,\infty)}^{\operatorname{U}}$ and any $U,V\subset [n]$, one has 
\begin{equation}\label{eq-useful-property}
    \Phi_{\tau}(U \cup V) \Phi_{\tau}(U \cap V)\leq \Phi_{\tau}(U) \Phi_{\tau}(V)\,. 
\end{equation}
This is because by definition we have
\begin{align*}
    \frac{\Phi_\tau(U)\Phi_\tau(V)}{\Phi_\tau(U\cap V)}=&\ \big(n^{1+4/d}d^{20}\big)^{|U|+|V|-|U\cap V|}\big(qd^6\big)^{\sum_{e\in \operatorname{K}_U}\tau_e+\sum_{e\in \operatorname{K}_V}\tau_e-\sum_{e\in \operatorname{K}_{U\cap V}}\tau_e}\\
    =&\ \big(n^{1+4/d}d^{20}\big)^{|U\cup V|}\big(qd^6\big)^{\sum_{e\in \operatorname{K}_{U}\cup \operatorname{K}_V}\tau_e}\ge\Phi_\tau(U\cup V)\,,
\end{align*}
where the inequality follows from the facts that $E(\operatorname{K}_U\cup \operatorname{K}_V)\subset E(\operatorname{K}_{U\cup V})$ and $0<qd^6<1$.

To prove {\eqref{eq-upp-bound-truncate-prob-proper-realiaztion}}, it is crucial to give more precise characterizations of improper realizations in $\mathcal E$. To get a feeling about this, recalling $\mathsf{1}_{\mathsf E}=(1,\dots,1)\in \{0,1\}^{\mathsf E}$, we claim that for any improper $\omega$ and any $\mathsf e\in \mathsf E$, there exists $W\subset [n]$ with $|W| \leq 2d^2$ such that \begin{equation}\label{eq-discussion-above-lemma-A-3}
    V(\mathsf e)\subset W\text{ and }\Phi_{\omega\oplus\mathsf{1}_{\mathsf E}\oplus\chi_S}(W)<(\log n)^{-1}\,. 
\end{equation}
In addition, if $\omega\in \mathcal E$, then it further holds that $|W|\le d^2$. In other words, for any $\omega\stackrel{*}{\sim}(\mathsf E,\chi_S)$ we have the following property: in the graph with edge indicators $\omega\oplus\mathsf{1}_{\mathsf E}\oplus\chi_S$, each edge $\mathsf e\in \mathsf E$ is contained in a bad subgraph with at most $d^2$ vertices. This is true because otherwise we have (note that $\Gc$ is increasing) 
$$
\mathbf{1}_{ \omega \oplus \chi_{\mathsf E }   \oplus \chi_S   \in \mathcal{G}   } = \mathbf{1}_{ \omega \oplus \chi^{\mathsf e}_{\mathsf E}   \oplus \chi_S   \in \mathcal{G}  }\,,\forall \chi_{\mathsf E}\in\{0,1\}^{\mathsf E}
$$
(where $\chi_{\mathsf E}^{\mathsf{e}}$ differs from $\chi_\mathsf E$ exactly for edge $\mathsf e$), and thus \eqref{eq-def-proper-realization} cancels to $0$, which contradicts the fact that $\omega$ is improper. We enhance this claim to a much stronger assertion, as incorporated in the next lemma.

\begin{lemma}{\label{lemma-containment-maximum-graph}}
For any $\chi_S\in \mathsf A_S$ and $\omega\stackrel{*}{\sim}({\mathsf E},\chi_S)$, there exists $W\subset [n]$ with $|W|\le d^2$ such that $V(\mathsf E)\subset W$ and $\Phi_{ \omega \oplus \mathsf{1}_{\mathsf E}  \oplus \chi_S  } (W)<(\log n)^{-1}$.
\end{lemma}

The proof of Lemma~\ref{lemma-containment-maximum-graph} is quite technical and it is postponed to Section~\ref{subsec-B2}. As a result of this lemma, we see that for any $\chi_S\in \mathsf A_S$, $\widehat\Pb[\omega\stackrel{*}{\sim}({\mathsf E},\chi_S)]$ is upper-bounded by
\[
\widehat\Pb\big[\omega:\exists W\subset [n]\text{ s.t. }|W|\le d^2,W\supset V(\mathsf E)\text{, and }\Phi_{\omega\oplus\mathsf{1}_{\mathsf E}\oplus \chi_S}(W)< (\log n)^{-1}\big]\,.
\]
In order to upper-bound the latter probability, a natural attempt is to employ the union bound. However, a direct union bound does not seem to yield an efficient upper bound due to the intricate intersection patterns of the induced graph on $W$ and $S$, which lead to additional complications. That being said, in the special scenario where the densest induced subgraph of $S$ is $S$ itself (to be precise, when $\Phi_{\omega\oplus\mathsf{1}_{\mathsf E}\oplus\chi_S}(V(S))\le \Phi_{\omega\oplus\mathsf{1}_{\mathsf E}\oplus\chi_S}(U)$ holds for any $U\subset V(S)$), a further reduction on the intersection pattern can be performed and then the union bound essentially works. Specifically, provided that $S$ satisfies the aforementioned condition, one has
\begin{align*}
&\qquad\{\omega:\exists W\text{ s.t. }|W|\le d^2,W\supset V(\mathsf E),\Phi_{\omega\oplus\mathsf{1}_{\mathsf E}\oplus\chi_S}(W)<(\log n)^{-1}\}\subset
\\
&\{\omega:\exists W'\text{ s.t. }|W'|\le 2d^2, W'\supset V(\mathsf E)\cup V(S),\Phi_{\omega\oplus\mathsf{1}_{\mathsf E}\oplus\chi_S}(W')<(\log n)^{-1}\}\,.
\end{align*}
The claim above could be easily verified by taking $W'=W\cup V(S)$ (where $W$ is the witness of the event on the left hand side) and then making use of the densest subgraph condition to show that $W'$ witnesses the event on the right hand side. Based on this, one may take the union bound over $W'$ to obtain an estimate that suffices for our purpose.

Inspired by the preceding observation, our main goal is to somehow reduce the problem to the case that $S$ itself is the densest induced subgraph, and this motivates the following lemma. Recall that we have assumed $|V(S)|\ge 10$ and $S$ contains all edges in $V(S) \setminus \mathsf E$. Hence $\mathsf{1}_{\mathsf E}\oplus \chi_S  $ encodes a graph with vertex set ${V(S)}$ (viewed as a weighted graph with weights $0$ or $1$), and thus for any subset $U$ of $V(S)$ we see that $\Phi_{ \mathsf{1}_{\mathsf E}   \oplus \chi_S   }(U)$ is well-defined (which is independent of the choice of $\omega$). 
\begin{lemma}{ \label{lemma-expand-graph} }
    For any $\chi_S\in\mathsf A_S$, there exists $\xi_S=\{\xi_e\}_{e\in E(S)}\in{[0,\infty)}^{E(S)}$ such that the following hold:  \\
    \noindent (a) $\xi_{e} \geq \chi_e,\forall e\in E(S)$.\\
    \noindent (b)  $\Phi_{ \mathsf{0}_{\mathsf E}   \oplus \xi_S   } (H) \geq (\log n)^{-1},\forall H\subset V(S)$ (recalling that $\mathsf{0}_{\mathsf E}=(0,\dots,0)\in \{ 0,1 \}^{\mathsf E}$).\\
    \noindent (c) $\Phi_{ \mathsf{1}_{\mathsf E}   \oplus \xi_S   } (V(S)) \leq \Phi_{ \mathsf{1}_{\mathsf E}   \oplus \xi_S   } (H),\forall H \subset V(S)$.
\end{lemma}
\begin{remark}
    Lemma~\ref{lemma-expand-graph} may be understood in the following way: it states that any self-admissible graph may be enlarged to a weighted graph (as in (a)) which is still self-admissible (in the sense of (b)), while its densest subgraph is itself (as in (c)). We also point out that it is plausible to impose the additional condition that $\xi_e\in \{0,1\}$ for all $e\in E(S)$, such that we will have a subgraph in the conventional sense. However, introducing the weighted graph does not causing any additional conceptual difficulty, while it helps circumvent several technical issues. Therefore, we opt to work with the weighted version.
\end{remark}

The proof of Lemma~\ref{lemma-expand-graph} is standard and we leave it to the end of Section~\ref{subsec-B2}. We now finish the proof of Lemma~\ref{lemma-prob-proper-realization}.
\begin{proof}[Proof of Lemma~\ref{lemma-prob-proper-realization}.]
    For any $\chi_S\in \mathsf A_S$ and $\omega\stackrel{*}{\sim}({\mathsf E},\chi_S)$, from Lemma~\ref{lemma-containment-maximum-graph} there exists $W\supset V(\mathsf E)$ such that $\Phi_{ \omega \oplus \mathsf{1}_{\mathsf E}   \oplus \chi_S   }(W)<(\log n)^{-1}$. Let $\xi_S$ be defined as in Lemma~\ref{lemma-expand-graph}. Then since $\chi_e\le \xi_e$ for all $e\in E(S)$ we have 
    \[
    \Phi_{\omega\oplus\mathsf{1}_{\mathsf E}\oplus\xi_S}(W)\le  \Phi_{\omega\oplus\mathsf{1}_{\mathsf E}\oplus\chi_S}(W)<(\log n)^{-1}\,.
    \]
    Using Lemma~\ref{lemma-expand-graph} (c), we have $$\Phi_{ \omega \oplus \mathsf{1}_{\mathsf E}   \oplus \xi_S } (W \cap V(S)) \geq \Phi_{ \omega \oplus \mathsf{1}_{\mathsf E} \oplus \xi_S }(V(S))\,.$$
    Combined with \eqref{eq-useful-property}, this yields that
    \begin{align*}
        \Phi_{ \omega \oplus \mathsf{1}_{\mathsf E} \oplus \xi_S } (W \cup V(S)) &\leq \frac{ \Phi_{ \omega \oplus \mathsf{1}_{\mathsf E} \oplus \xi_S }(W) \Phi_{ \omega \oplus \mathsf{1}_{\mathsf E} \oplus \xi_S }(V(S)) }{ \Phi_{ \omega \oplus \mathsf{1}_{\mathsf E} \oplus \xi_S }(W \cap V(S)) } \\
        &\leq \Phi_{ \omega \oplus \mathsf{1}_{\mathsf E}   \oplus \xi_S }(W)<(\log n)^{-1} \,.
    \end{align*}
    In conclusion, we have shown that once $\omega$ is an improper realization, there exists $W'=W\cup V(S) \subset [n]$ such that $W'\supset V(\mathsf E)\cup V(S)$, $|W'|\le 2d^2$ and $\Phi_{ \omega \oplus \mathsf{1}_{\mathsf E}  \oplus \xi_{S}  } (W')<(\log n)^{-1}$ (as outlined in the discussion below Lemma~\ref{lemma-containment-maximum-graph}). This implies that $\widehat\Pb[\omega\stackrel{*}{\sim}({\mathsf E},\chi_S)]$ is upper-bounded by
    \begin{align*}
    \widehat{\mathbb{P}}\big[\omega: \exists W \supset V(\mathsf E)\cup V(S),|W|\le 2d^2,\Phi_{\omega \oplus \mathsf{1}_{\mathsf E}  \oplus \xi_{S} }(W)<(\log n)^{-1} \big] \,.
    \end{align*}
    Denote $m=|V(S)|$ and $v=|V(\mathsf E) \setminus V(S)|$ for brevity. We will apply a union bound on all possible $W \subset [n]$ with $|W|=k\ge m+v$ such that $\Phi_{\omega \oplus \mathsf{1}_{\mathsf E}  \oplus \xi_{S} }(W)< (\log n)^{-1}$. For each $k\ge m+v$ we define
    \begin{equation}{ \label{eq-def-ENUM} }
        \mbox{ENUM}(k) = \Big\{ W \subset [n] : W\supset V(\mathsf E)\cup V(S), |W|= k \Big\} \,.
    \end{equation}
    Clearly for all $W \in \mbox{ENUM}(k)$, $|W \setminus (V(S) \cup V(\mathsf E))| = k-m-v$. Thus,
    \begin{align*}
        |\mbox{ENUM}(k)| = \binom{n}{k-m-v} \leq n^{k-m-v} \,.
    \end{align*}
    We also define
    \begin{equation}{ \label{eq-def-PROB} }
        \mbox{PROB}(k) = \max \Big\{ \widehat{\mathbb{P}}\big[\omega: \Phi_{ \omega \oplus \mathsf{1}_{\mathsf E}  \oplus \xi_{S}  }(W)<(\log n)^{-1} \big] : W \in \mbox{ENUM}(k) \Big\} \,.
    \end{equation}
    We next upper-bound $\mbox{PROB}(k)$. Denote $\|\xi_S\|_1= \sum_{e \in E(S)} \xi_e $. 
    In addition, define
    \begin{equation}{ \label{eq-def-tilde-E(k)} }
        \zeta(k) = \min \big\{ k' \geq 0 : \big(n^{1+4/d} d^{20}\big)^{k} \big(qd^6\big)^{k'+\|\xi_S\|_1+|\mathsf E|} < (\log n)^{-1} \big\} \,.
    \end{equation}
    Then there must be at least $\zeta(k)$ edges in $E(\operatorname{K}_W) \setminus (E(S) \cup \mathsf E)$, and thus
    \begin{align*}
        \widehat{\mathbb{P}} \big[ \Phi_{ \omega \oplus \mathsf{1}_{\mathsf E} \oplus\xi_{S} }(W)<(\log n)^{-1} \big] \leq \mathbb{P} \big[ \mathbf{B}(k^2,p) \geq \zeta(k) \big]\,.
    \end{align*}
    Due to the fact that $k\le 2d^2=n^{o(1)}$ and $p \le qd$ (since we have assumed $\rho \ge 1/d$ at the beginning of Section~\ref{sec-proof-of-sparse}), we see $k^2p=o(1)$ and thus the probability above is upper-bounded by $(k^2p)^{ \zeta(k)} \le (4d^5 q)^{\zeta(k)}$ from Poisson approximation.
    Moreover, using the conditions $\big(n^{1+4/d} d^{20}\big)^{k} \big(q d^6\big)^{\zeta(k)+\|\xi_S\|_1+|\mathsf E|} <(\log n)^{-1}$ by \eqref{eq-def-tilde-E(k)} and $\big(n^{1+4/d} d^{20}\big)^{m} \big(q d^6\big)^{\|\xi_S\|_1}\ge (\log n)^{-1}$ by Lemma~\ref{lemma-expand-graph} (b), we have 
    \begin{align*}
        q^{\zeta(k)} \leq \big(n^{1+4/d} d^{20}\big)^{-k+m} d^{-6(\zeta(k)+|\mathsf E|)}q^{-|\mathsf E|} \leq d^{-6\zeta(k)} \big(nd^{20}\big)^{-k+m} q^{-|\mathsf E|} \,.
    \end{align*}
    Thus, we have $\operatorname{PROB}(k) \leq \big(n d^{20}\big)^{-k+m}q^{-|\mathsf E|}$. By applying a union bound we know that
    \begin{align*}
    &\quad \widehat{\mathbb{P}}\big[\omega: \exists W \supset V(\mathsf E)\cup V(S), |W|\le 2d^2, \Phi_{\omega \oplus \mathsf{1}_{\mathsf E}  \oplus \xi_{S} }(W)<(\log n)^{-1} \big] \\
    \leq&\ \sum_{k=v+m}^{2d^2} |\mbox{ENUM}(k)| \cdot \mbox{PROB}(k)
        \leq \sum_{k=v+m}^{2d^2} n^{k-m-v} \big(n d^{20}\big)^{-k+m}q^{-|\mathsf E|} \\
        =&\ n^{-v} q^{-|\mathsf E|} \sum_{k=v+m}^{2d^2} d^{-20(k-m)}\le 2d^{-20v} n^{-v} q^{-|\mathsf E|}=2\big(n d^{20}\big)^{-|V(\mathsf E)\setminus V(S)| }q^{-|\mathsf E|} \,.
    \end{align*}
    This completes the proof of Lemma~\ref{lemma-prob-proper-realization}.
\end{proof} 

\subsection{Proof of Lemma~\ref{lemma-containment-maximum-graph} and Lemma~\ref{lemma-expand-graph}}\label{subsec-B2}

We now turn to the proof of Lemma~\ref{lemma-containment-maximum-graph}. Fix $\chi_S\in \mathsf A_S$ together with an improper realization $\omega \stackrel{*}{\sim}({\mathsf E},\chi_S)$. 
For each edge $\mathsf e\in \mathsf E$, 
we define $V_{\mathsf e}\subset [n]$ as
\begin{equation}{ \label{eq-def-V_k} }
    V_{\mathsf e} = \arg \min_{ \substack{ W\supset V(\mathsf e) \\ {|W| \leq 2d^2} } } \Phi_{ \omega \oplus \mathsf{1}_{\mathsf E}   \oplus \chi_S   } (W) \,,
\end{equation}
and if there are multiple minimizers, we arbitrarily choose one from those with the maximal size.
Recalling \eqref{eq-discussion-above-lemma-A-3}, it must hold that $\Phi_{ \omega \oplus \mathsf{1}_{\mathsf E}   \oplus \chi_S   }( V_{\mathsf e} )<(\log n)^{-1}$ for any $\mathsf e\in \mathsf E$. In particular, this implies $|V_{\mathsf e}|\le d^2$ for all $\mathsf e\in \mathsf E$ since $\omega\in \mathcal E$. The proof of Lemma~\ref{lemma-containment-maximum-graph} will be completed once we prove the following fact:
\begin{equation}\label{eq-containing}
V_{\mathsf e_1}\subset V_{\mathsf e_2}\text{ or }V_{\mathsf e_2}\subset V_{\mathsf e_1},\forall \mathsf e_1\neq\mathsf e_2\in \mathsf E\,.
\end{equation}
This is because given \eqref{eq-containing}, we can just take $W$ to be the maximal set among $\{V_{\mathsf e}:\mathsf e\in \mathsf E\}$.
\begin{proof}[Proof of \eqref{eq-containing}]
Suppose on the contrary that there exists $\mathsf e_1\neq \mathsf e_2\in \mathsf E$ such that $V_{\mathsf e_1}\not \subset V_{\mathsf e_2}$ and $V_{\mathsf e_2} \not \subset V_{\mathsf e_1}$. We first show that 
\begin{equation}{ \label{eq-worse-intersection} }
    \Phi_{ \omega \oplus \mathsf{1}_{\mathsf E}  \oplus \chi_S   } ( V_{\mathsf e_1} ) > \Phi_{ \omega \oplus \mathsf{1}_{\mathsf E}  \oplus \chi_S   } (V_{\mathsf e_1} \cap V_{\mathsf e_2}) \,.
\end{equation}
Otherwise, supposing \eqref{eq-worse-intersection} fails, we apply \eqref{eq-useful-property} and get that
\begin{align*}
    \Phi_{ \omega \oplus \mathsf{1}_{\mathsf E}   \oplus \chi_S   } (V_{\mathsf e_1} \cup V_{\mathsf e_2}) \leq \frac{ \Phi_{ \omega \oplus \mathsf{1}_{\mathsf E}   \oplus \chi_S   } (V_{\mathsf e_1})\Phi_{ \omega \oplus \mathsf{1}_{\mathsf E}   \oplus \chi_S   } (V_{\mathsf e_2}) }{ \Phi_{ \omega \oplus \mathsf{1}_{\mathsf E}   \oplus \chi_S   } (V_{\mathsf e_1} \cap V_{\mathsf e_2}) } \le \Phi_{ \omega \oplus \mathsf{1}_{\mathsf E}   \oplus \chi_S   } (V_{\mathsf e_2}) \,.
\end{align*}
But $|V_{\mathsf e_2}|<|V_{\mathsf e_1}\cup V_{\mathsf e_2}|\le 2d^2$, contradicting the choice of $V_{\mathsf e_2}$. Similarly we can show that 
\begin{align}
    \Phi_{ \omega \oplus \mathsf{1}_{\mathsf E}  \oplus \chi_S   } ( V_{\mathsf e_2} ) > \Phi_{ \omega \oplus \mathsf{1}_{\mathsf E}  \oplus \chi_S   } (V_{\mathsf e_1} \cap V_{\mathsf e_2}) \,. \label{eq-useful-property-symmetry-part}
\end{align}

An immediate consequence of \eqref{eq-worse-intersection} and \eqref{eq-useful-property-symmetry-part} is that neither $V(\mathsf e_1)$ nor $V(\mathsf e_2)$ is contained in $ V_{\mathsf e_1} \cap V_{\mathsf e_2}$ (otherwise it would violate the choice of $V_{\mathsf e_1}$ or $V_{\mathsf e_2}$). Now we define
$
\mathsf{E}'=\{\mathsf{e}\in \mathsf E:V(\mathsf{e}) \subset V_{\mathsf e_1} \cap V_{\mathsf e_2}\}
$ (it is possible that $\mathsf E'=\emptyset$), 
and we see that $\mathsf e_1,\mathsf e_2\notin \mathsf E'$.
For $\chi_{\mathsf E'}\in \{0,1\}^{\mathsf E'}$, we say that $\chi_{\mathsf E'}$ is locally admissible with respect to $\omega$ and $\chi_S$, if the induced subgraph on $V_{\mathsf e_1}\cap V_{\mathsf e_2}$ of the graph encoded by $\omega\oplus\chi_{\mathsf E'}\oplus\mathsf{1}_{\mathsf E\setminus \mathsf E'}\oplus\chi_S$ is admissible, where $\mathsf{1}_{\mathsf E\setminus \mathsf E'}$ denotes for the all-one vector in $\{0,1\}^{\mathsf{E}\setminus\mathsf{E}'}$. 

The key point of deriving a contradiction lies in the following claim:

\begin{claim}{\label{Lemma-no-bad-graph}}
For any $\chi_{\mathsf E'}$ which is locally admissible with respect to $\omega$ and $\chi_S$, let $G(\chi_{\mathsf E'})$ be the graph encoded by $\omega \oplus \chi_{\mathsf{E}'}\oplus\mathsf{1}_{\mathsf E\setminus \mathsf E'}   \oplus \chi_S  $. Then there is no bad subgraph of $G(\chi_{\mathsf E'})$ with at most $d^2$ vertices that contains $\mathsf e_1$.
\end{claim}
\begin{proof}[Proof of Claim~\ref{Lemma-no-bad-graph}]
Suppose the statement is false for some $\chi_{\mathsf E'}$ locally admissible with respect to $\omega$ and $\chi_S$.
Define (note that $V(\mathsf e_1)$ below means the collection of the endpoints for the edge $\mathsf e_1$, which is drastically different from $V_{\mathsf e_1}$)
\begin{align}
    \widetilde{V}_{\mathsf e_1} = \arg \min_{ \substack{ V(\mathsf e_1) \subset W \\ |W| \leq 2d^2 } } \Phi_{ \omega \oplus \chi_{\mathsf E'}\oplus\mathsf{1}_{\mathsf E\setminus \mathsf E'}   \oplus \chi_S   } (W)\,,
\end{align}
and again if there are multiple minimizers, we arbitrarily choose one that maximizes the size of $\widetilde{V}_{\mathsf e_1}$.
Then it holds that 
\begin{equation}\label{eq-<}
    \Phi_{\omega\oplus\chi_{\mathsf E'}\oplus\mathsf{1}_{\mathsf E\setminus \mathsf E'}\oplus\chi_S}(\widetilde V_{\mathsf e_1})<(\log n)^{-1}
\end{equation} by our assumption. This implies $|\widetilde V_{\mathsf e_1}|\le d^2$, since otherwise we have $d^2<|\widetilde V_{\mathsf e_1}|\le 2d^2$ and $$\Phi_{\omega\oplus\mathsf {1}_{\mathsf{E}}\oplus\chi_S}(\widetilde V_{\mathsf e_1})\le \Phi_{\omega\oplus\chi_{\mathsf E'}\oplus\mathsf{1}_{\mathsf E\setminus \mathsf E'}\oplus \chi_{S}}(\widetilde V_{\mathsf e_1})<(\log n)^{-1}\,,$$
contradicting our assumption that $\omega\stackrel{*}{\sim}(\mathsf E,\chi_S)$. 

We now prove that $\widetilde{V}_{\mathsf e_1} \subset V_{\mathsf e_1}$. By our choice of $\widetilde{V}_{\mathsf e_1}$ we have
\begin{align*}
    \Phi_{ \omega \oplus \chi_{\mathsf E'}\oplus\mathsf{1}_{\mathsf E\setminus \mathsf E'}   \oplus \chi_S   } (\widetilde{V}_{\mathsf e_1}) \le \Phi_{ \omega \oplus \chi_{\mathsf E'}\oplus \mathsf{1}_{\mathsf E\setminus \mathsf E'}   \oplus \chi_S   } ( \widetilde{V}_{\mathsf e_1} \cap V_{\mathsf e_1} ) \,.
\end{align*}
Thus (using $V(\mathsf e) \subset V_{\mathsf e_1} \cap V_{\mathsf e_2} \subset V_{\mathsf e_1}$ for $\mathsf e\in \mathsf E'$),
\begin{align*}
    \Phi_{ \omega \oplus \mathsf{1}_{\mathsf E}   \oplus \chi_S   } ( \widetilde{V}_{\mathsf e_1} ) =&\ \Phi_{ \omega \oplus \chi_{\mathsf{E}'}\oplus\mathsf{1}_{\mathsf E\setminus \mathsf E'}   \oplus \chi_S   } ( \widetilde{V}_{\mathsf e_1} ) \cdot \big(qd^6\big)^{ \#\{ \mathsf e\in \mathsf E':V(\mathsf e) \subset \widetilde{V}_{\mathsf e_1}, \chi_{\mathsf e} = 0 \} }  \\
    \leq&\ \Phi_{ \omega \oplus \chi_{\mathsf E'}\oplus\mathsf{1}_{\mathsf E\setminus \mathsf E'}   \oplus \chi_S   } ( \widetilde{V}_{\mathsf e_1} \cap V_{\mathsf e_1} ) \cdot \big(qd^6\big)^{ \#\{ \mathsf e\in \mathsf E':V(\mathsf e) \subset \widetilde{V}_{\mathsf e_1}, \chi_{\mathsf e} = 0  \}}  \\
    =&\ \Phi_{ \omega \oplus \chi_{\mathsf{E}'}\oplus \mathsf{1}_{\mathsf E\setminus \mathsf E'}   \oplus \chi_S   } ( \widetilde{V}_{\mathsf e_1} \cap V_{\mathsf e_1} ) \cdot \big(qd^6\big)^{ \#\{  \mathsf e\in \mathsf E':V(\mathsf e) \subset \widetilde{V}_{\mathsf e_1}\cap V_{\mathsf e_1}, \chi_{\mathsf e} = 0  \} }  \\
    =&\ \Phi_{ \omega \oplus \mathsf{1}_{\mathsf E} \oplus \chi_S } (\widetilde{V}_{\mathsf e_1} \cap V_{\mathsf e_1}) \,.
\end{align*}
Combined with \eqref{eq-useful-property}, this implies that
\begin{align*}
    \Phi_{ \omega \oplus \mathsf{1}_{\mathsf E} \oplus  \chi_S  } ( V_{\mathsf e_1} \cup \widetilde{V}_{\mathsf e_1} ) \leq \frac{ \Phi_{ \omega \oplus \mathsf{1}_{\mathsf E} \oplus \chi_S } (V_{\mathsf e_1}) \Phi_{ \omega \oplus \mathsf{1}_{\mathsf E}   \oplus \chi_S } (\widetilde{V}_{\mathsf e_1}) }{ \Phi_{ \omega \oplus \mathsf{1}_{\mathsf E} \oplus \chi_S  } (V_{\mathsf e_1} \cap \widetilde{V}_{\mathsf e_1}) } \leq \Phi_{ \omega \oplus \mathsf{1}_{\mathsf E} \oplus \chi_S } (V_{\mathsf e_1}) \,,
\end{align*}
which implies $\tilde V_{e_1} \subset V_{e_1}$ by the choice of $V_{e_1}$ and the fact that $|V_{e_1} \cup \tilde V_{e_1}| \leq 2d^2$. 

Having proved that $\widetilde V_{\mathsf e_1}\subset V_{\mathsf e_1}$, we note that $\widetilde{V}_{\mathsf e_1}\cap V_{\mathsf e_2}=\widetilde{V}_{\mathsf e_1} \cap (V_{\mathsf e_1} \cap V_{\mathsf e_2})$ induces an admissible graph in $\omega \oplus \chi_{\mathsf E'}\oplus \mathsf{1}_{\mathsf E\setminus \mathsf E'}   \oplus \chi_S  $ (by the definition of local admissibility). Thus, 
\begin{align}
    \Phi_{ \omega \oplus \chi_{\mathsf E'}\oplus \mathsf{1}_{\mathsf E\setminus \mathsf E'}  \oplus \chi_S   } (\widetilde{V}_{\mathsf e_1} \cap V_{\mathsf e_2})\ge (\log n)^{-1} \,.
    \label{eq->}
\end{align}
In addition, using $V(\mathsf e) \subset V_{\mathsf e_1} \cap V_{\mathsf e_2}$ for $\mathsf e\in \mathsf E'$, we get that
\begin{align}
    \nonumber\frac{ \Phi_{ \omega \oplus \mathsf{1}_{\mathsf E} \oplus  \chi_S } (\widetilde{V}_{\mathsf e_1}) }{ \Phi_{ \omega \oplus \chi_{\mathsf E'}\oplus\mathsf{1}_{\mathsf E\setminus \mathsf E'} \oplus  \chi_S } (\widetilde{V}_{\mathsf e_1}) } &= \big(qd^6\big)^{ \# \{ \mathsf e\in \mathsf E' : V(\mathsf e) \subset \widetilde{V}_{\mathsf e_1}, \chi_{\mathsf e} = 0\} } = \big(qd^6\big)^{ \# \{ \mathsf e\in \mathsf E': V(\mathsf e) \subset \widetilde{V}_{\mathsf e_1} \cap V_{\mathsf e_2}, \chi_{\mathsf e}=0 \} } \\
    &= \frac{ \Phi_{ \omega \oplus  \mathsf{1}_{\mathsf E}  \oplus  \chi_S  } (\widetilde{V}_{\mathsf e_1} \cap V_{\mathsf e_2}) }{ \Phi_{ \omega \oplus \chi_{\mathsf E'}\oplus \mathsf{1}_{\mathsf E\setminus \mathsf E'} \oplus \chi_S  } (\widetilde{V}_{\mathsf e_1} \cap V_{\mathsf e_2}) } \,.\label{eq-=}
\end{align}    
Combining \eqref{eq-=} with \eqref{eq-<}, \eqref{eq->} yields that
\begin{align*}
    \frac{ \Phi_{ \omega \oplus \mathsf{1}_{\mathsf E}   \oplus \chi_S   } (\widetilde{V}_{\mathsf e_1}) }{ \Phi_{ \omega \oplus \mathsf{1}_{\mathsf E}   \oplus \chi_S   } (\widetilde{V}_{\mathsf e_1} \cap V_{\mathsf e_2}) } = \frac{ \Phi_{ \omega \oplus \chi_{\mathsf E'}\oplus \mathsf{1}_{\mathsf E\setminus \mathsf E'}  \oplus \chi_S   } (\widetilde{V}_{\mathsf e_1}) }{ \Phi_{ \omega \oplus \chi_{\mathsf E'}\oplus \mathsf{1}_{\mathsf E\setminus \mathsf E'}  \oplus \chi_S   } (\widetilde{V}_{\mathsf e_1} \cap V_{\mathsf e_2}) } <\frac{(\log n)^{-1}}{(\log n)^{-1}} = 1 \,. 
\end{align*}
Thus, from \eqref{eq-useful-property} (keep in mind that $\widetilde V_{\mathsf e_1}\cap (V_{\mathsf e_1}\cap V_{\mathsf e_2})=\widetilde V_{\mathsf e_1}\cap V_{\mathsf e_2}$),
\begin{align*}
    \Phi_{ \omega \oplus \mathsf{1}_{\mathsf E}   \oplus \chi_S   } ( \widetilde{V}_{\mathsf e_1} \cup (V_{\mathsf e_1} \cap V_{\mathsf e_2}) ) & \leq \frac{ \Phi_{ \omega \oplus \mathsf{1}_{\mathsf E}   \oplus \chi_S   }(V_{\mathsf e_1} \cap V_{\mathsf e_2}) \cdot \Psi_{ \omega \oplus \mathsf{1}_{\mathsf E}  \oplus \chi_S   }(\widetilde{V}_{\mathsf e_1}) }{ \Phi_{ \omega \oplus \mathsf{1}_{\mathsf E}   \oplus \chi_S   }(\widetilde{V}_{\mathsf e_1} \cap V_{\mathsf e_2}) }  \\
    &<\Phi_{ \omega \oplus \mathsf{1}_{\mathsf E}   \oplus \chi_S   }(V_{\mathsf e_1} \cap V_{\mathsf e_2}) \overset{\eqref{eq-worse-intersection}}{<} \Phi_{ \omega \oplus \mathsf{1}_{\mathsf E}   \oplus \chi_S   }(V_{\mathsf e_1}) \,,
\end{align*}
which contradicts the choice of $V_{\mathsf e_1}$.
\end{proof}
With the claim in hand, we now derive a contradiction (under the assumption that \eqref{eq-containing} fails as supposed at the beginning) as follows. Firstly, in the summation of \eqref{eq-def-proper-realization}, to make it possible for $\omega \oplus \chi_{\mathsf E}   \oplus \chi_S   \in \mathcal{G}$ where $\chi_{\mathsf E}=\chi_{\mathsf E'}\oplus\chi_{\mathsf E\setminus \mathsf E'}$, one must choose $\chi_{\mathsf E'} $ to be locally admissible. Hence the left hand side of \eqref{eq-def-proper-realization} equals 
\begin{align*}
    \sum_{  \chi_{\mathsf E'}\ \textup{locally admissible}  } \Big( \sum_{ \chi_{\mathsf E\setminus \mathsf E'} \in \{  0,1\}^{\mathsf E\setminus \mathsf E'} }  (-1)^{\# \{ \mathsf e\in \mathsf E':\chi_{\mathsf e}=0 \}} \mathbf{1}_{ \omega \oplus \chi_{\mathsf E},\chi_{\mathsf E\setminus \mathsf E'} \oplus \chi_S   \in \mathcal{G}   } \Big) \,.
\end{align*}
For any $\chi_{\mathsf E\setminus \mathsf E'}\in \{0,1\}^{\mathsf E\setminus \mathsf E'}$, let $\chi_{\mathsf E\setminus \mathsf E'}^1\in \{0,1\}^{\mathsf E\setminus \mathsf E'}$ be defined as $\chi^1_{\mathsf e_1}=1-\chi_{\mathsf e_1}$ and $\chi_{\mathsf e}^1=\chi_{\mathsf e}$ for $\mathsf e\neq \mathsf e_1$. For each fixed locally admissible $\chi_{\mathsf E'}$, we claim that
\begin{align}\label{eq-cancelation}
    \mathbf{1}_{ \omega \oplus \chi_{\mathsf E'}\oplus\chi_{\mathsf E\setminus \mathsf E'}  \oplus \chi_S   \in \mathcal{G}   } = \mathbf{1}_{ \omega \oplus \chi_{\mathsf E'}\oplus \chi_{\mathsf E\setminus \mathsf E'}^1   \oplus \chi_S   \in \mathcal{G}   } \,,\forall \chi_{\mathsf E\setminus \mathsf E'}\in \{0,1\}^{\mathsf E\setminus \mathsf E'}\,.
\end{align}
Otherwise $\mathsf e_1$ would be pivotal for $\Gc$ under some realization $\omega\oplus\chi_{\mathsf E'}\oplus\chi_{\mathsf E\setminus \mathsf E'}\oplus\chi_S$ and thus there must be a bad subgraph with size at most $d^2$ containing $\mathsf e_1$. This would also give rise to a bad subgraph with no more than $d^2$ vertices in $\omega\oplus\chi_{\mathsf E'}\oplus\mathsf{1}_{\mathsf E\setminus \mathsf E'}\oplus\chi_S$ which contradicts Claim~\ref{Lemma-no-bad-graph}. Having verified
\eqref{eq-cancelation}, we see that the left hand side of \eqref{eq-def-proper-realization} cancels to $0$, contradicting the fact that $\omega$ is an improper realization. This implies that \eqref{eq-containing} is true and thus finishes the proof of Lemma~\ref{lemma-containment-maximum-graph}.
\end{proof}

To complete the proof of Proposition~\ref{prop-another-bound-conditional-expectation}, we are left with the proof of Lemma~\ref{lemma-expand-graph}.
\begin{proof}[Proof of Lemma~\ref{lemma-expand-graph}]
    If $E(S)=\emptyset$ (i.e. all the edges in $K_{V(S)}$ belongs to $\mathsf E$), we simply take $\xi_S=\chi_S$, and thus (a), (b) trivially hold and (c) is also true since we have assumed $|V(S)| \geq 10$. We now assume $E(S) \neq \emptyset$. We consider an arbitrary order on $E(S)$ and we will define $\xi_S$ inductively as follows: suppose that have we defined $\xi_S$ on $\mathsf F$ as $\xi_{\mathsf F}$, and suppose that $\mathtt e$ is the minimal edge in $E(S) \setminus \mathsf F$. We then define
    \begin{equation}
        \Xi_{\mathtt e}=\Big\{ \xi \geq 0 : \forall \ W \subset V(S), \Phi_{ \mathsf{0}_{\mathsf E} \oplus \xi \mathsf{1}_{\mathtt e} \oplus \xi_{\mathsf F} \oplus \chi_{E(S)\setminus(\mathsf F\cup \{\mathtt e\})}   }(W) \ge (\log n)^{-1} \Big\} \,,\label{equ-def-lambda^(i)_k}
    \end{equation}
    and set $\xi_{\mathtt e}=\sup{\Xi_{\mathtt e}}$. By repeating the procedure above, we claim that it will end up with some $\xi_S\in [0,\infty)^{E(S)}$ satisfying the desired properties.
    
    Denote $\mathtt e_0$ as the minimal element in $E(S)$. Initially, by the fact that $\chi_S\in\mathsf A_S$, we see $\chi_{\mathtt e_0}\in \Xi_{\mathtt e_0}$. In addition, by taking $W=V(\mathtt e_0)$ we see $\Xi_{\mathtt e_0}$ is bounded. Also it is clearly closed,
   and thus we conclude that $\Xi_{\mathtt e_0}$ is a compact set which contains $\chi_{\mathtt e_0}$, so $\xi_{\mathtt e_0}\in \Xi_{\mathtt e_0}\subset [0,\infty)$ and satisfies $\xi_{\mathtt e_0}\ge \chi_{\mathtt e_0}$. Given this,
    it can be easily shown by induction that in any step of the procedure, $\Xi_{\mathtt e}$ is a compact set which contains $\chi_{\mathtt e}$, so $\xi_{\mathtt e}\in \Xi_{\mathtt e}\subset  [0,\infty)$ and $\xi_{\mathtt e}\ge \chi_{\mathtt e}$. This suggests that $\xi_S=(\xi_{e})_{e\in E(S)}\in [0,\infty)^{E(S)}$ satisfies Item (a), and we also note that it satisfies Item (b) by the definition of the $\Xi_{\mathtt e}$ for the final edge $\mathtt e$.

    It remains to verify that $\xi_S$ defined as above satisfies Item (c). To this end, we define for each $\mathtt e\in E(S)$,
    \begin{align*}
        W_{\mathtt e} = \arg \min \Big\{ \Phi_{ \mathsf{0}_{\mathsf{E}} \oplus \xi_S  } (W) : V(\mathtt e) \subset W \subset V(S) \Big\} \,,
    \end{align*}
    and if there are multiple minimizers, we arbitrarily choose one among those with the maximal number of vertices. For each $\mathtt e\in E(S)$, denote $\mathsf F_{\mathtt e}=\{\mathtt e'\in E(S):\mathtt e'\preceq\mathtt e\}$. Then by the definition of $\xi_{\mathtt e}$, we see
    \[
    (\log n)^{-1}=\Phi_{\mathsf{0}_{\mathsf E} \oplus\xi_{\mathsf F_{\mathtt e}}\oplus\chi_{E(S)\setminus\mathsf F_{\mathtt e}} }(W_{\mathtt e})\ge \Phi_{\mathsf{0}_{\mathsf E} \oplus\xi_S }(W_{\mathtt e})\ge (\log n)^{-1}\,,
    \]
    and thus $\Phi_{\mathsf{0}_{\mathsf E} \oplus\xi_S }(W_{\mathtt e})=(\log n)^{-1}$ for each $\mathtt e\in E(S)$. In addition, by Item (b) (which we have already verified) it holds that for any $\mathtt e_1,\mathtt e_2\in E(S)$, $\Phi_{\mathsf{0}_{\mathsf E} \oplus\xi_S }(W_{\mathtt e_1}\cap W_{\mathtt e_2})\ge (\log n)^{-1}$, $\Phi_{\mathsf{0}_{\mathsf E} \oplus\xi_S }(W_{\mathtt e_1}\cup W_{\mathtt e_2})\ge (\log n)^{-1}$. Hence by \eqref{eq-useful-property},
    \begin{align*}
    (\log n)^{-2}\le&\ \Phi_{\mathsf{0}_{\mathsf E} \oplus\xi_S }(W_{\mathtt e_1}\cup W_{\mathtt e_2}){\Phi_{\mathsf{0}_{\mathsf E} \oplus\xi_S }(W_{\mathtt e_1}\cap W_{\mathtt e_2})}\\
    \le&\ {\Phi_{\mathsf{0}_{\mathsf E} \oplus
    \xi_S }(W_{\mathtt e_1})\Phi_{\mathsf{0}_{\mathsf E} \oplus\xi_S }(W_{\mathtt e_2})}=(\log n)^{-2}\,,
    \end{align*}
    and thus the equality must hold. Then from our choice of $W_{\mathtt e},\mathtt e\in E(S)$, we have $W_{\mathtt e_1}=W_{\mathtt e_1}\cup W_{\mathtt e_2}=W_{\mathtt e_2}$. Since this is true for any $\mathtt e_1,\mathtt e_2\in E(S)$, we conclude that all the $W_{\mathtt e}$'s are the same, and the common set must be $V(S)$ since it contains $V(\mathtt e)$ for each $\mathtt e\in E(S)$ and $S$ has no isolated vertex (also recalling $W_{\mathtt e}\subset V(S)$ by definition). As a result, we have for any $H\subset V(S)$, it holds that
    \[
    \Phi_{\mathsf{0}_{\mathsf E} \oplus\xi_S }(V(S))=(\log n)^{-1} \le \Phi_{\mathsf{0}_{\mathsf E} \oplus\xi_S }(H)\,.
    \]
    This yields
    \begin{align*}
    \Phi_{\mathsf{1}_{\mathsf E} \oplus\xi_S }(V(S)) =\Phi_{\mathsf{0}_{\mathsf E} \oplus\xi_S }(V(S)) \big(qd^6\big)^{|\mathsf E|} 
    \leq \Phi_{\mathsf{0}_{\mathsf E} \oplus\xi_S }(H) \big(qd^6\big)^{|\mathsf E|} \le \Phi_{\mathsf{1}_{\mathsf E} \oplus\xi_S }(H)\,,
    \end{align*}
    which verifies Item (c) and thus completes the proof.
\end{proof}

\subsection{Verification of \ref{eq-quanlitative-estimation}.}\label{subsec-B3}
Finally, we detail the verification of \eqref{eq-quanlitative-estimation}. When $4(|E(S_1)|+|E(S_2)|-2|E(S_0)|) \leq 5(|V(S_1)|+|V(S_2)|-2 |V(S_0)|)-1$, we have (recall that $d\ge 100$)
    \begin{align*}
        &\big(n d^{20}\big)^{-\frac{|V(S_1)|+|V(S_2)|}{2}+|V(S_0)|} 2^{2+|E(S_1)|+|E(S_2)|-2|E(S_0)|} \\
        \leq &\ 4d^{-2}\cdot n^{-\frac{|V(S_1)|+|V(S_2)|}{2}+|V(S_0)|} \left({d^8}/{2}\right)^{-(|E(S_1)|+|E(S_2)|-2|E(S_0)| )} \\
        \le &\ n^{-\frac{|V(S_1)|+|V(S_2)|}{2}+|V(S_0)|} d^{-7( |E(S_1)|+|E(S_2)|-2|E(S_0)| )} \,.
    \end{align*}
    And when $4(|E(S_1)|+|E(S_2)|-2|E(S_0)|) \geq 5( |V(S_1)|+|V(S_2)|-2|V(S_0)| )$, we have (note that $q=n^{-1+o(1)}$ and so $4q\le n^{-4/5}d^{-14}$ for large $n$)
    \begin{align*}
        &(4q)^{ \frac{|E(S_1)|+|E(S_2)|}{2}-2|E(S_0)|}\leq ( n^{4/5}d^{14})^{ -\frac{|E(S_1)|+|E(S_2)|}{2}-2|E(S_0)|} \\
        \leq &\ n^{ -\frac{2}{5}( |E(S_1)|+|E(S_2)|-2|E(S_0)| ) } d^{-7( |E(S_1)|+|E(S_2)|-2|E(S_0)| )} \\
        \leq &\ n^{-\frac{|V(S_1)|+|V(S_2|)}{2}+|V(S_0)|} d^{-7( |E(S_1)|+|E(S_2)|-2|E(S_0)| )}\,.
    \end{align*}
This verifies \eqref{eq-quanlitative-estimation} as desired.
\bibliographystyle{plain}

\end{document}